\documentclass[showpacs,onecolumn,aps,prx,long bibliography,superscriptaddress,notitlepage]{revtex4-1}
\usepackage[utf8]{inputenc}
\usepackage{amsmath, amsthm, amssymb, bm}
\usepackage{mathrsfs}
\usepackage{extarrows}
\usepackage{pifont}
\usepackage{diagbox}
\usepackage{xcolor}
\usepackage{multirow}
\usepackage{graphicx}
\usepackage{hyperref}
\usepackage{float}
\hypersetup{colorlinks=true, linkcolor=blue, citecolor=blue, urlcolor=black}
\usepackage{caption}
\usepackage{braket}
\usepackage{pgfplots}
\usepackage{multirow}

\usepackage{tikz}

\usepackage{qcircuit}
\usepackage{physics}
\usepackage{array}
\usepackage{bbm}

\usepackage{caption}
\newtheorem{lemma}{Lemma}
\newtheorem{theorem}{Theorem}

\newcommand{\ii}{\mathsf{i}}
\newcommand{\polylog}{\mathrm{polylog}}
\newcommand{\alpl}{\alpha_{\{\gamma_i\gamma_j\}, 2t+1}^{L}}
\newcommand{\alpr}{\alpha_{\{\gamma_i\gamma_j\}, 2t+1}^{R}}

\newcommand{\vabs}[1]{\left\| #1 \right\|}

\newcommand{\supket}[1]{|#1 \rangle\rangle}
\newcommand{\supbra}[1]{\langle\langle #1 |}
\newcommand{\cbra}[1]{\{ #1 \}}
\newcommand{\pbra}[1]{\left( #1 \right)}
\newcommand{\sbra}[1]{\left[ #1 \right]}
\newcommand{\supketbra}[2]{
    \supket{#1 } \supket{#1 } \supbra{#2} \supbra{#2} 
}
\newcommand{\supbraket}[1]{\langle\langle #1 \rangle\rangle}

\newcommand{\tra}[1]{\text{Tr}\left( #1 \right)}
\newcommand{\Ord}[1]{\mathcal{O}\left( #1 \right)}

\newcommand{\floor}[1]{\lfloor #1 \rfloor}
\newcommand{\mean}{\mathop{\mathbb{E}}}
\newcommand{\Var}{\mathop{\mathrm{Var}}}

\newcommand{\Ecal}{\mathcal{E}}

\newcommand{\Rcal}{\mathcal{R}}
\newcommand{\Scal}{\mathcal{S}}
\newcommand{\Tcal}{\mathcal{T}}
\newcommand{\Ucal}{\mathcal{U}}
\newcommand{\Mcal}{\mathcal{M}}

\newcommand{\Ubb}{\mathbb{U}}
\newcommand{\Mbb}{\mathbb{M}}
\newcommand{\Ebb}{\mathbb{E}}

\begin{document}
\title{Adaptive-depth randomized measurement for fermionic observables}
\author{Kaiming Bian}
%\email{12231295@mail.sustech.edu.cn}
\affiliation{Shenzhen Institute for Quantum Science and Engineering,
Southern University of Science and Technology, Shenzhen 518055, China}
\affiliation{International Quantum Academy, Shenzhen 518048, China}
\affiliation{Guangdong Provincial Key Laboratory of Quantum Science and Engineering, Southern University of Science and Technology, Shenzhen, 518055, China}
\author{Bujiao Wu}
\email{bujiaowu@gmail.com}
\affiliation{Shenzhen Institute for Quantum Science and Engineering,
Southern University of Science and Technology, Shenzhen 518055, China}
\affiliation{International Quantum Academy, Shenzhen 518048, China}
\affiliation{Guangdong Provincial Key Laboratory of Quantum Science and Engineering, Southern University of Science and Technology, Shenzhen, 518055, China}

\begin{abstract}
Accurate estimation of fermionic observables is essential for advancing quantum physics and chemistry. The fermionic classical shadow (FCS) method offers an efficient framework for estimating these observables without requiring a transformation into a Pauli basis. However, the random matchgate circuits in FCS require polynomial-depth circuits with a brickwork structure, which presents significant challenges for near-term quantum devices with limited computational resources.
To address this limitation, we introduce an adaptive-depth fermionic classical shadow (ADFCS) protocol designed to reduce the circuit depth while maintaining the estimation accuracy and the order of sample complexity. 
Through theoretical analysis and numerical fitting, we establish that the required depth for approximating a fermionic observable $H$ scales as $\max\cbra{d^2_{\text{int}}(H)/\log n, d_{\text{int}}(H)}$ where $d_{\text{int}}$ is the interaction distance of $H$.
We demonstrate the effectiveness of the ADFCS protocol through numerical experiments, which show similar accuracy to the traditional FCS method while requiring significantly fewer resources. Additionally, we apply ADFCS to compute the expectation value of the Kitaev chain Hamiltonian, further validating its performance in practical scenarios. 
Our findings suggest that ADFCS can enable more efficient quantum simulations, reducing circuit depth while preserving the fidelity of quantum state estimations, offering a viable solution for near-term quantum devices.
\end{abstract}
\maketitle

\section{Introduction}

The simulation of strongly correlated fermionic systems is a fundamental application of quantum computing, playing a crucial role in advancing our understanding of quantum materials and chemical processes~\cite{lieb1961two, si1996kosterlitz,aspuru2005simulated,whitfield2011simulation,auerbach2012interacting,cao2019quantum}.
These systems, characterized by complex electron correlations and dynamics, demand precise computational techniques to capture their behavior. 
A key computational challenge is the calculation of expectation values, such as energy levels in fermionic Hamiltonians, which serve as essential descriptors of system properties. 
Quantum algorithms such as variational quantum eigensolver (VQE)\cite{tilly2022variational, cerezo2021variational} and quantum Monte Carlo methods\cite{foulkes2001quantum, carlson2015quantum} provide powerful tools for this task. However, they typically require a significant number of measurements to achieve accurate results.

The classical shadow (CS) algorithm~\cite{huang2020predicting} was introduced to give a classical estimator for a quantum state $\rho$. The CS algorithm is particularly suited for estimating linear properties of $\rho$, such as the expectation values $\cbra{\tra{\rho Q_i}}_{i=1}^m$. The number of measurements required is $\Ord{\max_i \vabs{Q_i}^2_{\text{shadow}} \log m / \varepsilon^2}$, where $\varepsilon$ is the desired estimation error and $\vabs{\cdot}_{\text{shadow}}$ denotes the shadow norm, which depends on the unitary ensemble. 
However, for certain local fermionic observables, such as $\gamma_1\gamma_{2n}$, the shadow norm scales exponentially, rendering the CS algorithm inefficient. 
To address this limitation, fermionic classical shadow (FCS) algorithms~\cite{bonet2020nearly, Zhao21Fermionic, low2022classical, wan2022matchgate, zhao2024group,heyraud2024unified} were developed. These algorithms primarily differ from the CS algorithm by using a Gaussian unitary ensemble rather than the Clifford group for randomness. Both Gaussian unitary and Clifford elements require polynomial-size quantum circuits, which become challenging for near-term quantum devices due to issues like gate noise and limited coherence time~\cite{patel2008optimal, jiang2020optimal, jiang2018quantum, arute2019quantum, Wu21Strong, cao2023generation}. In response to this, several approaches have been proposed to design shallow-depth CS protocols~\cite{bertoni2024shallow, schuster2024random, rozon2024optimal,ma2024construct}. However, for the FCS side,
Zhao et al.~\cite{Zhao21Fermionic} have proved that there does not exist subgroup $G\subseteq \text{Cl}_n \cap \Mbb_n$ which has a better fermionic shadow norm, where Cl$_n$ is the $n$-qubit Clifford group and $\Mbb_n$ is the matchgate group.  
King et al.~\cite{king2024triply} also demonstrated that single-copy measurements require $\Omega\pbra{n^{k/2}/\epsilon^2}$ copies of $\rho$ in any classical shadow based protocol to achieve an $\epsilon$-error estimation for $k$-local Majorana strings. 
An intriguing open question remains:
\begin{center}
  \textit{Is there a shallow-depth FCS algorithm that is efficient for some specific set of fermionic observables and maintains the same sample complexity as the original FCS algorithm?}  
\end{center}

In this work, we identify the minimum circuit depth required by the FCS protocol for $k$-local Majorana strings.
 Beginning with 2-local Majorana strings, we derive an expression that characterizes the dependence of sample complexity on circuit depth. Building on these findings, we investigate the depth required to minimize sample complexity for general $k$-local Majorana strings.

Specifically, we propose an adaptive-depth fermionic classical shadow (ADFCS) algorithm for optimizing the depth of the brickwork random matchgate circuit, as shown in Fig.~\ref{fig:sketch}(a-b). 
Similar to the FCS algorithm, ADFCS applies a random matchgate circuit $U_{Q_{d^\ast}}$ and utilizes classical shadows to estimate the expectation value of an observable $H$. 
By Ref.~\cite{Jiang18Quantum}, constructing a Haar-random matchgate circuit with a brickwork architecture requires $\Omega(n^2)$ depth. In contrast, ADFCS adaptively selects the depth $d^\ast$ based on the interaction distance of the observable and constructs the random matchgate circuit $U_{Q_{d^\ast}}$ accordingly.
The depth $d^\ast$ is linear to $\max\cbra{d^2_{\text{int}}(H)/\log n, d_{\text{int}}(H)}$ where $d_{\text{int}}$ is the interaction distance of the fermionic observable $H$. This relationship is obtained by modeling the variance as a random walk in response to variations in $d$, and selecting an appropriated $d^\ast$ such that the summation of the random walk probability on certain sites is large enough. 

For example, we numerically demonstrate that ADFCS with a depth of $d^\ast = 3$ achieves comparable estimation precision to the FCS method when estimating the expectation value of the Kitaev chain Hamiltonian for $n=10$ qubits. Additionally, we evaluate the ADFCS method by analyzing the estimation error for several $k$-local  Majorana strings with varying interaction distances. The numerical results indicate that adaptively selecting a relatively shallow depth can achieve a sample complexity similar to that of FCS for most $k$-local Majorana strings.

\begin{figure}
    \centering
    \includegraphics[width=0.9\linewidth]{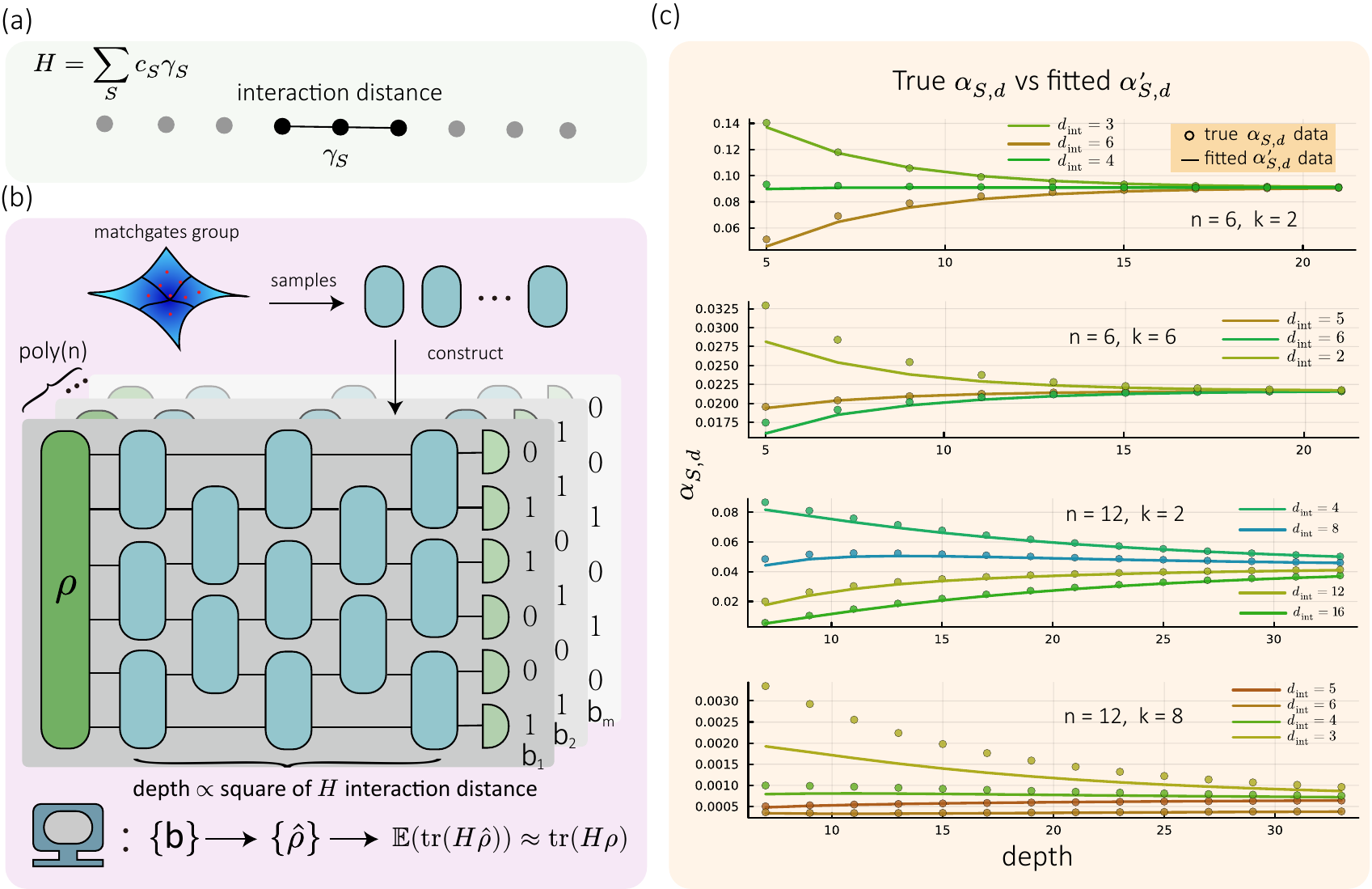}
    \caption{\centering Illustration of Adaptive-Depth Fermionic Classical Shadow (ADFCS).
(a) The interaction distance of the observable $H$. The Majorana operator $\gamma_S$ in $H$ can be interpreted as the interaction between particles. 
The interaction distance $d_\text{int}(H)$ is defined as the maximum interaction distance of each operator $\gamma_S$ in $H$.
(b) Sketch of the ADFCS protocol. Unitary gates are independently sampled from the group of two-qubit matchgates within a brickwork architecture. The measured bitstrings $b$ are processed on a classical computer to generate the classical representation of $\rho$, and $\mathrm{Tr}(H \rho)$ can be estimated using the generated classical estimator $\hat{\rho}$. The depth of the random circuit is determined by the fermionic Hamiltonian $H = \sum c_S \gamma_S$.
 (c) Illustration of $\alpha'_{S,d}$ and $\alpha_{S,d}$ with the increase of depth $d$ for different number of qubits $n$ and Majorana strings $\gamma_S$. 
 } 
    \label{fig:sketch}
\end{figure}

\section{Preliminary}

\noindent\textbf{Notations.} 
 A non-zero linear operator $\supket{A}$ can be vectorized as $\supket{A}:=\frac{A}{\sqrt{\tra{A^\dagger A}}}$ where the normalization ensures that it is properly scaled. 
Specifically, for the Majorana observable
$\gamma_S$, the vectorized form is
$\supket{\gamma_S}:=\frac{\gamma_S}{\sqrt{2^n}}$, where $n$ is the number of qubits. Moreover, the operator $A$ can be represented by the Pauli basis, which is known as the Pauli-transfer matrix (PTM) representation. 
The inner product between the vectorized formation of two operators $A,B$ is defined as $\supbraket{A|B} :=\frac{\tra{A^\dagger B}}{\sqrt{\tra{A^\dagger A} \tra{B^\dagger B}}}$. The operation of superoperator $\Ecal$ can be denoted as $\Ecal\supket{A}:=\frac{\Ecal(A)}{\sqrt{\tra{A^\dagger A}}}$.

\vspace{8pt}
\noindent\textbf{Majorana operators and matchgate circuits.}
The Majorana operators can be defined as 
\begin{equation}
    \gamma_{2 j-1}:=a_j+a_j^{\dagger}, \quad \gamma_{2 j}:=-\ii\left(a_j-a_j^{\dagger}\right),
    \label{eq: standard Majorana operators}
\end{equation}
where $a_j,a_j^\dagger$ are the annihilation and creation operators for the $j$-th site respectively. 
Observables in fermionic systems can be expressed as the linear combination of the products of Majorana operators~\cite{hackl2021bosonic}
\begin{equation}
    \gamma_S := \gamma_{i_1}\gamma_{i_2}\cdots\gamma_{i_{|S|}},
    \label{eq: gammaS}
\end{equation}
where the cardinality $|S|=\#\cbra{i \mid i\in S}$,  $i_1 < {i_2}< \cdots < i_{|S|}$, and $S = \{i_1 , {i_2}, \cdots , i_{|S|}\}$. Especially, we set $\gamma_\emptyset$ be the identity operator $\mathbbm{1}$ when $S$ is the empty set $\emptyset$. For simplicity, we refer to a Majorana string $\gamma_S$ with $|S| = k$ as a \emph{$k$-local Majorana string} throughout this manuscript.
Let $d_{\text{int}}(S)$ denote the interaction distance of the Majorana string $\gamma_S$, defined as
\begin{equation}
\label{eq: dint def}
d_{\text{int}}(S) := \max\{\,|i_{j+1} - i_{j}| \mid i_j \in S,\, j \in [2n-1]\,\}.
\end{equation}
The interaction distance of a fermionic operator $H = \sum_{S\in\Scal} \alpha_{S} \gamma_{S}$ is defined as $d_{\text{int}}(H) = \max_{S\in \Scal} d_{\text{int}}(S)$.
A Majorana operator can be represented in Pauli formation as $\gamma_{2 j-1}=\left(\prod_{i=1}^{j-1} Z_i\right) X_j, \ \ \gamma_{2 j}=\left(\prod_{i=1}^{j-1} Z_i\right) Y_j$ for any $j\in[n]$ via Jordan-Wigner transformation.

Except for the standard Majorana operators in Eq.~\eqref{eq: standard Majorana operators}, there are different bases to express the fermionic Hamiltonian. 
Different bases can be obtained by the orthogonal transformation of the standard Majorana operators $\tilde{\gamma}_\mu=\sum_{\nu=1}^{2 n} Q_{\mu \nu} \gamma_\nu,$
where $Q$ is a real orthogonal matrix, $Q \in O(2n)$. 
The different sets of bases $\{\tilde{\gamma}_\mu\}$ are also called the Majorana operators because they satisfy the same anti-commutation relation $\{\tilde{\gamma}_\nu, \tilde{\gamma}_\mu\} = 2\delta_{\mu\nu}$. 
The orthogonal transformations of Majorana operators can be represented by unitaries
\begin{equation}
 U_Q^{\dagger} \gamma_\mu U_Q=\sum_{\nu=1}^{2 n} Q_{\mu \nu} \gamma_\nu = \tilde{\gamma}_\mu,
 \label{eq:matchgate_transform}
\end{equation}
where the unitaries $U_Q$ are called fermionic Gaussian unitaries. Followed by Eq.~\eqref{eq:matchgate_transform}, the fermionic Gaussian unitary preserves the cardinality $|S|$ of $\gamma_S$, which can be expressed as $ U_Q^{\dagger} \gamma_S U_Q=\sum_{S' \in {[2n]\choose |S|}} \det\pbra{Q|_{S,S'}}\gamma_{S'}.$ The set ${[2n] \choose |S|}$ represents the collection of all subsets of $\{1, 2, \dots, 2n\}$ that contain exactly $|S|$ elements, and the matrix $Q|_{S,S'}$ is obtained from the matrix $Q$ by selecting the rows indexed by the set $S$ and the columns indexed by the set $S'$.

Matchgate circuits are the quantum circuit representation of fermionic Gaussian unitaries~\cite{wan2022matchgate}. 
The collection of all $n$-qubit matchgate circuits constitutes the matchgate group $\Mbb_n$. This group is generated by rotations of the form $\exp\pbra{\ii \theta X_{\mu} X_{\mu+1}}, \exp\pbra{\ii \theta Z_{\nu}}$ and $X_\eta$, where $\theta$ is a real parameter, $\nu,\eta \in [n]$, $\mu\in[n-1]$. 
Notably, operations within the matchgate group can be efficiently simulated on a classical computer. This efficiency arises because the action of a matchgate circuit $U_{Q}$ corresponds to Givens rotations associated with an orthogonal matrix $Q$, facilitating polynomial-time classical simulation~\cite{valiant2002quantum,jiang2018quantum}.

We present the twirling of the matchgate group in the following lemma, as it is utilized in the proof of our main results.
\begin{lemma}[The three moments of uniform distribution in $\Mbb_n$, Ref.~\cite{wan2022matchgate}]
\label{lemma: 1}
The $k$-moment twirling $\Ecal^{(j)}$ is defined by
\begin{equation}
    \Ecal^{(j)}(\cdot):=\int_{U_Q \in \Mbb_n} \dd U_Q ~ U_Q^{\otimes j} (\cdot ) U_Q ^{\dagger \otimes j}.
\end{equation}
The first three moments are
\begin{equation}
    \begin{aligned}
        \Ecal^{(1)} &= \supket{\mathbbm{1}}\supbra{\mathbbm{1}} \\
        \Ecal^{(2)} &= \sum_{k=0}^{2n} {2n\choose k}^{-1}\sum_{\substack{S,S'\subseteq [2n] \\ |S| = |S'| = k}} \supket{\gamma_S}\supket{\gamma_S}\supbra{\gamma_{S'}}\supbra{\gamma_{S'}} \\
        \Ecal^{(3)} &= \sum_{
\substack{ k_1,k_2,k_3 \geq 0 \\ k_1 + k_2+k_3\leq 2n }}\supket{\Rcal}\supbra{\Rcal} ,
    \end{aligned}
\end{equation}
where 
\begin{equation}
    \supket{\Rcal}={2n \choose k_1,k_2,k_3, 2n - k_1-k_2-k_3}^{-1/2} \sum_{\substack{A,B,C \text{disjoint}\\
|A|=k_1,|B|=k_2, |C| = k_3
}
}\supket{\gamma_{A} \gamma_{B}}\supket{\gamma_{B} \gamma_{C}} \supket{\gamma_{C} \gamma_{A}}.
\end{equation}

\label{lem:FCS_shadow_channel}
\end{lemma}

 \vspace{8pt}
\noindent\textbf{(Fermionic) classical shadow protocol.}

The classical shadow (CS) process involves randomly applying a unitary $U$ from some unitary ensemble $\Ubb$ (such as the Clifford group) on a quantum state $\rho$, followed by measurement in computational bases, yielding outcomes $\supket{b}$ that form a classical description of $\rho$, known as ``\emph{classical shadow}'', denoted as $\supket{\hat{\rho}} = \mathcal{M}^{-1}_{\text{CS}}(\Ucal^{-1}\supket{b})$, where $\Ucal$ is the superoperator of $U$, and $\mathcal{M}_{\text{CS}}:=\Ebb_{\Ucal\in\Ubb}\sbra{\Ucal^{\dagger}\sum_{b} \supket{b}\supbra{b} \Ucal}$. By choosing $\Ubb$ as Clifford group Cl$_n$, the classical shadow channel can be simplified to $\Mcal_{\text{CS}} = \Pi_0 + \frac{1}{2^n + 1}\Pi_1$, where $\Pi_0$ is the projector onto the identity subspace and $\Pi_1$ is the projector onto the subspace spanned by all Pauli bases except identity~\cite{huang2020predicting, chen2021robust}. This simple representation allows efficient construction of the classical shadow $\hat{\rho}$.

In the framework of FCS, the shadow channel $\mathcal{M}_{\text{FCS}}:=\Ebb_{Q\in O(2n)}\sbra{\Ucal_{Q}^{\dagger}\sum_{b} \supket{b}\supbra{b} \Ucal_Q}$
is not invertible across the entire Hilbert space.  Its invertibility is constrained to a subspace spanned by even operators, denoted as
$ \Gamma_{\text{even}} := \mathrm{span}\{ \gamma_S \mid |S| = 2j, j=0,1,2,\cdots, n\}.$
For set $S$ with odd cardinality, the channel $\mathcal{M}_{\text{FCS}}(\gamma_S) = 0$, rendering these components inaccessible.

As a result, FCS is limited to recovering expectation values of observables that reside within the $\Gamma_{\text{even}}$ subspace. However, this restriction is sufficient for capturing physical observables, as they are typically even operators due to the conservation of fermionic parity~\cite{turner2011topological}. Unless otherwise specified, all discussions and calculations are assumed to take place within the $\Gamma_{\text{even}}$.

\section{Adaptive-depth fermionic classical shadows}

Here, we consider the ADFCS by selecting the unitary ensemble as $d$-depth local matchgates with brickwork architecture, as shown in Fig.~\ref{fig:sketch}(b), which is widely applied in superconducting quantum devices~\cite{arute2019quantum}. 
It has been shown that certain observables cannot be addressed by replacing the global Cl$_n \cap \Mbb_n$ group with shallow-depth local matchgate circuits~\cite{Zhao21Fermionic}. However, it remains open when global elements are necessary and how to efficiently generate the expectation values of a given set of fermionic observables using the shallowest depth matchgate circuits while minimizing the number of samplings. 
Here, we address this question by exploring the relationship between the number of required samples and the depth of the matchgate circuit for some specific set of fermionic observables using the brickwork architecture. 
{Our approach introduces a $d$-depth ADFCS channel to estimate the expectation value of a fermionic observable $H$. The estimation of the expectation value is divided into the estimation of a set $\{\gamma_S\}$. The optimal depth order is determined by the maximum interaction distance of the set $\{\gamma_S\}$. }

\subsection{Shadow channel analysis in ADFCS}
\label{sec: shadow channel}
For simplicity, we utilize $U_{Q_d}$ to denote a $d$-depth matchgate circuit with brickwork structure, and $\Ubb_{Q_d}$ to denote the set of all $d$-depth matchgate circuit throughout this manuscript. We denote $\Mcal_d$ as the $d$-depth ADFCS shadow channel, which maps a quantum state $\rho$ to
\begin{equation}
    \mathcal{M}_d(\rho) := \mathop{\mathbb{E}}\limits_{U\in\Ubb_{Q_d},b\in\cbra{0,1}^n} \left[ \bra{b}U_{Q_d} 
  \rho U_{Q_d}^\dagger \ket{b}  U_{Q_d}^\dagger \ket{b} \bra{b}U_{Q_d}\right].
\end{equation}
Due to the expectation property of the subset $\Ubb_{Q_d}$, the Majorana operator remains invariant under the action of $\Mcal_{d}$ up to a scaling factor $\alpha_{S,d}$, i.e.,
\begin{equation}
\mathcal{M}_d(\gamma_S) = \alpha_{S,d} \gamma_S,
\label{eq: lemma1 eigen}
\end{equation}
where
$\alpha_{S,d} := \int \dd U_{Q_d} \abs{\bra{\bm 0}U_{Q_d} \gamma_S U_{Q_d}^\dagger\ket{\bm 0}}^2.$
The proof is shown in Appendix~\ref{appendix: 1}.

First, we show that the ADFCS can unbiasedly rebuild the state $\rho$, thereby providing an unbiased estimation of the expectation value. Directly follows from the linearity of the shadow channel $\mathcal{M}_d$, we have
\begin{equation}
    \mathop{\mathbb{E}}\limits_{U_{Q_d}, b } ( \mathcal{M}^{-1}_{d}(U_{Q_d}^\dagger \ket{b} \bra{b}U_{Q_d})) = \mathcal{M}^{-1}_{d}(\mean(U_{Q_d}^\dagger \ket{b} \bra{b}U_{Q_d}) ) = \rho
    \label{eq: unbias}
\end{equation}
when $\alpha_{S,d}$ is non-zero.

Eq.~\eqref{eq: unbias} demonstrates that the ADFCS protocol provides an unbiased estimation. We now analyze the required sample complexity to achieve high precision.  
For any observable $\gamma_S$ with respect to the average of quantum states $\rho$, the variance of the ADFCS estimator $v = \tra{\hat{\rho} \gamma_S}$ is bounded as follows:
\begin{equation}
\Var[v] \leq \frac{1}{\alpha_{S,d}}.
\label{eq: variance and alpha}
\end{equation}
The detailed derivation of Eq.~\eqref{eq: variance and alpha} is provided in Appendix~\ref{appendix: bound variance}. Using this result and applying Chebyshev's inequality, the estimation error $\abs{v - \tra{\rho \gamma_S}}$ can be bouned to $\varepsilon$ using $\mathcal{O}(\frac{1}{\alpha_{S,d}\varepsilon^2})$ copies of the quantum states with high success probability. Numerical results demonstrate that the variance is stable across various quantum states. In the following sections, we will present the calculation for $\alpha_{S,d}$, along with its {lower bound.}

\subsection{Tensor network approach to variance bound}

By Eq.~\eqref{eq: variance and alpha}, bounding the variance requires calculating $\alpha_{S,d}$. Here, we demonstrate that $\alpha_{S,d}$ can be represented as a tensor network and computed through contraction.
Using the PTM representation, $\alpha_{S,d}$ can be expressed as
\begin{equation}
\label{eq: alpha tensor def}
\alpha_{S,d} = \supbra{\bm 0,\bm 0}  \int \dd \Ucal_{Q_d} \mathcal{U}_{Q_d}^{\otimes 2} \supket{\gamma_S, \gamma_S},
\end{equation}
where $\mathcal{U}_{Q_d}$ is the superoperator of ${U}_{Q_d}$, $\mathcal{U}_{Q_d} \supket{\gamma_S} = U_{Q_d} \gamma_S U_{Q_d}^\dagger$.
Consequently, the integration over $\Ucal_{Q_d}$ can be broken down into a product of a series of integrations over the independent 2-qubit matchgates. 
By Lemma \ref{lem:FCS_shadow_channel}, we can represent the integration over each 2-qubit matchgate as a fourth-order tensor $\Tcal$ with indices $\sigma_1,\sigma_2, \sigma_3, \sigma_4$. The explicit form of this tensor is given by
\begin{equation}
\label{eq: T}
    \Tcal^{\sigma_1\sigma_2}_{~~\sigma_3\sigma_4} := \supbra{\sigma_1, \sigma_2}\supbra{\sigma_1, \sigma_2}\int_{Q\sim O(4)} \dd \Ucal_{Q} \mathcal{U}_{Q}^{\otimes 2} \supket{\sigma_3, \sigma_4}\supket{\sigma_3, \sigma_4},
\end{equation}
where the indices $\sigma_1, \sigma_2, \sigma_3, \sigma_4$ of tensor $\Tcal$ represent the Pauli operators. The tensor $\Tcal$ is the integration of one blue random gate in Fig.~\ref{fig:sketch}. 
By connecting these tensors $\Tcal$ in the same brickwork architecture as $ U_{Q_d} $, the integration over the superoperator in Eq.~\eqref{eq: alpha tensor def} can be transformed into a tensor network. The transformation details are shown in Appendix \ref{appendix 3}.

\begin{figure}
    \centering
    \includegraphics[width=0.8 \linewidth]{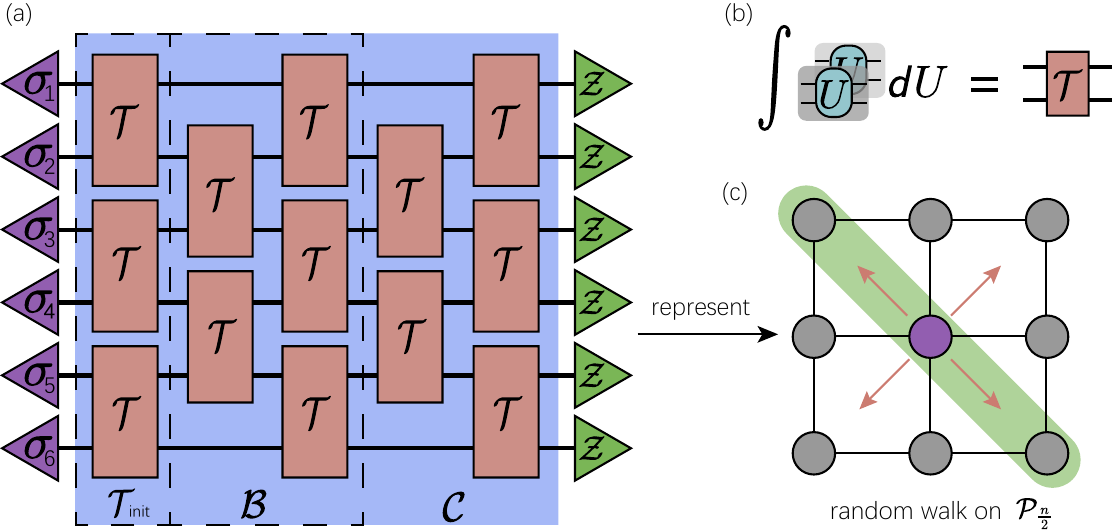}
    \caption{\centering Representation of the tensor contraction of $\mathcal{C}$ as a random walk on the polynomial space $\mathcal{P}_{\frac{n}{2}}$. (a) The tensor network used to compute $\alpha_{S,d}$. The purple triangles represent Pauli operators corresponding to $\gamma_S$ via the Jordan-Wigner transformation, while the green triangles, denoted as $\mathcal{Z}$, correspond to the supervector $\frac{1}{2} (\supket{\mathbbm{1}, \mathbbm{1}} + \supket{Z, Z})$. (b) The definition of $\Tcal$ in Eq.~\eqref{eq: T}. (c) The contraction of $\mathcal{C}$ onto the polynomial space $\mathcal{P}_{\frac{n}{2}}$. Transitioning the state $\supket{\gamma_S, \gamma_S}$ and contracting it with $\mathcal{Z}$ is equivalent to summing the probabilities over the diagonal sites in a 2D random walk.
    }
    \label{fig: figure4}
\end{figure}

The value of $\alpha_{S,d}$ can be calculated by contracting the corresponding brickwork tensor network, which is shown in Fig.~\ref{fig: figure4}(a). Specifically, the brickwork tensor network, denoted as $\mathcal{C}$, is contracted with $\supbra{\bm 0,\bm 0}$ and $\supket{\gamma_S, \gamma_S}$ in the PTM representation to calculate $\alpha_{S,d}$
\begin{equation}
    \alpha_{S,d} = \supbra{\bm 0,\bm 0} \mathcal{C}  \supket{\gamma_S, \gamma_S}.
    \label{eq: alpha Twhole}
\end{equation}
Following the proof of Lemma 5 in Ref.~\cite{bertoni2024shallow}, 
the tensor network $\mathcal{C}$ can be represented by a matrix product operator with bond dimension $2^{\mathcal{O}(d)}$. It suggests that calculating $\alpha_{S,d}$ by tensor network contraction is efficient only when the depth $d$ is shallow.

\subsection{Adaptive depth for two-local  Majorana strings}

Here, we start by determining the optimal depth order for the random matchgate circuit $U_{Q_d}$, such that the estimation of a set of 2-local Majorana strings $\cbra{\gamma_S|\mid \abs{S}=2}$ requires the same sample complexity compared to the full matchgate group $\Mbb_n$. 
Notably, for specific sets of Majorana strings $\cbra{\gamma_S}$, the optimized number of measurements can be significantly improved compared to FCS when using a smaller $d$. 
As an instance, suppose we need to estimate the expectation of $\gamma_{i}\gamma_{i+1}$. The random measurement in FCS actually shuffles the $\gamma_{i}\gamma_{i+1}$ to any 
$\gamma_{j} \gamma_k$ uniformly. However, a matchgate circuit with $d=2$ will only slightly perturb the   $\gamma_{i}\gamma_{i+1}$ to the nearby operators $\{\gamma_{j}\gamma_{k} \mid i-2\leq j \leq i+2, i-1\leq k \leq i+3\}$. A smaller sample space allows it to approach its theoretical mean with fewer samples compared to the FCS method.

We will investigate how the sample complexity of the quantum state changes as the depth of the matchgate circuits with brickwork structures increases. To determine the required depth of a random matchgate circuit for efficient measurement of a specific set of $2$-local Majorana strings $\cbra{\gamma_S}$, we derive an expression for $\alpha_{S,d}$ and identify the optimal depth $d^\ast$ such that $\alpha_{S,d^\ast} = \mathcal{O}(\frac{1}{\mathrm{poly}(n)})$. 
By restricting the input state of the whole tensor network $\mathcal{C}$ to the subspace $\Gamma_2$, we can simplify the tensor contraction into a polynomial space. This isomorphism introduces more refined structures and allows us to analyze the system more effectively.

We will focus on random matchgate circuits characterized by an even number of qubits and an odd circuit depth. For the random matchgate circuits with this structure, the contraction of $\mathcal{C}$ can be represented within the space of quadratic polynomials with $ \frac{n}{2} $ variables. For the circuits with even circuit depth, the contraction can be analyzed using a similar method. Denote the quadratic polynomial representation space as $\mathcal{P}_{\frac{n}{2}}$, the contraction of $\mathcal{C}$  can be described by the random walk with a simple pattern in space $\mathcal{P}_{\frac{n}{2}}$.

To represent this random walk, we use the transition process $\mathcal{B}$ to denote a pair of successive gate layers starting from an odd layer (see Fig.~\ref{fig: figure4}(a)),
\begin{equation}
    \mathcal{B} \supket{\gamma_S, \gamma_S} = \sum_{S':|S'| = |S|} \mathrm{Prob}(\gamma_{S'}|\gamma_S) \supket{\gamma_{S'}, \gamma_{S'}},
    \label{eq: T random walk}
\end{equation}
% which is shown in Fig.~\ref{fig: figure4}(a). 
The contraction of $\mathcal{C}$ can be understood as the repeated application of $ \mathcal{B} $ to the state $ \supket{\gamma_S, \gamma_S} $. 
As previously mentioned, we focus on the circuit with odd depth. Then, the entire tensor network can be expressed as
\begin{equation}
    \mathcal{C} = \mathcal{B}^{\lfloor d/2 \rfloor} \, \mathcal{T}_{\text{init}},
\end{equation}
where $\Tcal_\text{init} = \Tcal^{\otimes \frac{n}{2}}$ (and the circuit can be expressed as $\mathcal{C} = \mathcal{B}^{ d/2 }$ if the depth is an even number). Each $ \mathcal{B} $ corresponds to a single transition step in a random walk. 
Therefore, the entire tensor network $\mathcal{C}$ represents a random walk {starting from the site corresponds to $ \gamma_S $ and taking the transition} $ \lfloor \frac{d}{2}\rfloor $ steps, which is shown in Fig.~\ref{fig: figure4}(c). 
The details of mapping the tensor contraction of $\alpha_{S,d}$ to the random walk in $\mathcal{P}_{\frac{n}{2}}$ is shown in Appendix \ref{sec: Mapping the action of tensors to random walk}.

We observe that the random walk follows the pattern of a symmetry lazy random walk (SLRW)~\cite{giuggioli2020exact, lawler2010random} almost everywhere within the {polynomial space $\mathcal{P}_{\frac{n}{2}}$.}
SLRW is a stochastic process in which the walker has a probability of staying in the same position at each step and equal probabilities of moving to neighboring positions. 
Firstly, we consider the random walk transition $\mathcal{B}$ across the subspace $\Gamma_2$ adheres to the SLRW
\begin{equation}
    \mathcal{L}_{\Gamma_2}\supket{\gamma_S, \gamma_S} = \sum_{S':|S'| = |S|} \mathrm{Prob}_L(\gamma_{S'}|\gamma_S) \supket{\gamma_{S'}, \gamma_{S'}},
\end{equation}
where $\mathcal{L}_{\Gamma_2}$ is the transition operator in SLRW, and the concrete form of $\mathcal{L}_{\Gamma_2}$ is shown in Appendix \ref{sec: mapping the action of tensors to random walk}, Eq.~\eqref{eq: appendix L gamma def}. We replace the transition operator $\mathcal{B}$ with the SLRW transition operator $\mathcal{L}_{\Gamma_2}$ in $\mathcal{C}$ and use the transition $\mathcal{L}_{\Gamma_2}$ to calculate the tensor network contraction
\begin{equation}
    \alpha_{S,d}^L = \supbra{\bm 0,\bm 0} \mathcal{L}_{\Gamma_2}^{\lfloor \frac{d}{2}\rfloor} \Tcal_\text{init} \supket{\gamma_S, \gamma_S}.
    \label{eq: alphal Twhole}
\end{equation}
As we have shown in Table \ref{table: transition table of teto} of Appendix \ref{sec: Mapping the action of tensors to random walk}, the transition result of $\mathcal{L}_{\Gamma_2}$ equals the results $\mathcal{B}$ for most states in $\Gamma_2$, this indicates that the variable $\alpha_{S,d}^L$ likely constitutes the dominant contribution to $\alpha_{S,d}$. Additionally, numerical fitting further confirms that $\alpha_{S,d}$ is primarily governed by $\alpha_{S,d}^L$, as shown in Fig. \ref{fig:sketch} (b).

There is a known analytical propagation equation for the SLRW~\cite{giuggioli2020exact}. By applying this propagation equation and making certain approximations, the $ \alpha_{S,d} $ can ultimately be expressed as a Poisson summation
\begin{equation}
    \alpha_{S,d}^L = \frac{1}{3\sqrt{\pi (d-1)}} \sum_{k = -\infty}^\infty \left(e^{-\frac{(kn+a)^2}{d-1}} + e^{-\frac{(kn+b)^2}{d-1}}  \right) + \mathcal{O}(e^{-\pi^2 d}).
    \label{eq: analytical form of alpha}
\end{equation}
In Eq.~\eqref{eq: analytical form of alpha}, we denote the Majorana observable $\gamma_S$ with $\abs{S} = 2$
as $\gamma_i \gamma_j$, the variable $a$, $b$ are defined by $a:= |\lfloor\frac{i-1}{4}\rfloor-\lfloor\frac{j-1}{4}\rfloor|$ , and $b := \lfloor\frac{i-1}{4}\rfloor+\lfloor\frac{j-1}{4}\rfloor+1$. We leave the detailed proof of Eq.~\eqref{eq: analytical form of alpha} 
to Appendix \ref{sec: estimate the order of alpl}.

We aim to estimate the order of $\alpha_{S,d}$ by $\alpha_{S,d}^L$, which can be achieved by bounding the ratio $\alpha_{S,d} /\alpha_{S,d}^L$ to a constant value. 
To bound the ratio, we introduce an auxiliary function
\begin{equation}
\label{eq: def del}
    \Delta(S,S',d) = \supbra{\gamma_{S' }, \gamma_{S'}} 
\mathcal{B}^{\lfloor \frac{d}{2}\rfloor} \Tcal_\text{init} -  \mathcal{L}_{\Gamma_2}^{\lfloor \frac{d}{2}\rfloor} \Tcal_\text{init} \supket{\gamma_S, \gamma_S},
\end{equation}
where $S'\in\Gamma_2$. 
Recall that the operator $\gamma_{S'}$ is defined as $\gamma_{S'} = \gamma_{u} \gamma_{v}$ where $S' = \cbra{u,v}$. 
In numerical experiments, we observe that the term $\Delta(S, S',d)$ is negative if and only if the element $u$ is close to $v$ for $S'=\cbra{u,v}$ and any $S,d$, as shown in Fig.~\ref{fig: appendix assumption} of Appendix \ref{appendix relation between alphaL and alpha}. Based on this finding, we show that the ratio $\alpha_{S,d} /\alpha_{S,d}^L$ can be bounded by a constant value. The concrete description for bounding the ratio $\alpha_{S,d} /\alpha_{S,d}^L$ is provided in Appendix \ref{appendix relation between alphaL and alpha}.

The relation between $\alpha^L_{S,d}$ and $\alpha_{S,d}$, named $\alpha_{S,d}\geq c\alpha^L_{S,d}$ for positive constant $c$, indicated that the order of $\alpha_{S,d}^L$ consistent with the order of $\alpha_{S,d}$. Here, we propose the adaptive depth $d^\ast$ of a random matchgate circuit by analyzing the order of $\alpha^L_{S,d}$. 
By Eq.~\eqref{eq: analytical form of alpha}, for $ a = \mathcal{O}(\mathrm{log}(n)) $ and the measurement depth $ d^\ast = \mathcal{O} (\log n)$ can yield a polynomially small $ \alpha_{S,d^*} = \Omega\left(\frac{1}{n} \right)$ up to a log factor, we postpone the proof in Appendix~\ref{appendix efficiency}. 
{When $a=\omega\pbra{\log(n)}$, the remainder error term $\mathcal{O}(e^{-\pi^2 d})$ in Eq.~\eqref{eq: analytical form of alpha} scales polynomially small. }
Thus, we adaptively select the depth $d$ based on the first term of $\alpha_{S,d}^L$, which is decided by the interaction distance $d_{\text{int}} = |i-j|$,
\begin{equation}
    d^\ast = \Theta\left( \max\left\{\frac{d_{\text{int}}^2}{\log(n)}, ~d_{\text{int}}  \right\} \right).
\end{equation}
We show that the $\alpha_{S,d}^L$ scales as $\mathcal{O}(1/\text{poly}(n))$ when the depth is selected as $d^\ast$ in Appendix \ref{sec: estimate the order of alpl} Lemma \ref{corollary: 1}. Combined with the relationship between $\alpha^L_{S,d}$ and $\alpha_{S,d}$, it suggests that the sample complexity scales polynomially with the adaptively selected depth $d^\ast$.

\subsection{Adaptive depth for general $k$-local Majorana strings}

Here, we focus on extending adaptive depth to general $k$-local Majorana strings for constant $k$. Fermionic observables generated by $k$-local Majorana strings for constant $k$ play a critical role in various quantum models across physics and chemistry. These operators are essential in describing physical observables such as interaction terms, electron correlation, and pairing mechanisms. They appear in models like quantum chemistry Hamiltonians, generalized Hubbard models, spin chain models, etc~\cite{lieb1961two, si1996kosterlitz,aspuru2005simulated,whitfield2011simulation,auerbach2012interacting,cao2019quantum}.

The required circuit depth depends on the interaction distance $ d_{\text{int}}$ of $ \gamma_S $.
Recall Eq.~\eqref{eq: gammaS} that the elements in $S$ are a monotonically increasing sequence. Thus, the $ d_{\text{int}} $ can be interpreted as the maximum distance between adjacent elements in $S$. 
We prove that a random measurement circuit $\mathcal{U}_{Q_d}$ with $d=\Theta(\log n)$ is sufficient to make $ \alpha_{S,d} $ as large as $ \mathcal{O}(1/\text{poly}(n)) $ if $d_{\text{int}} = \mathcal{O}(\log n)$ and $k$ is a constant. 
From Eq.~\eqref{eq: alpha Twhole} and Eq.~\eqref{eq: T random walk}, we conclude that the value of $ \alpha_{S,d} $ is the sum of probabilities associated with specific sites after the random walk. When the distance of the near neighborhood $i_{2j-1}, i_{2j} \in S$ logarithmically small $i_{2j}- i_{2j-1} = \mathcal{O}(\log n)$, a logarithmically transition steps can traverse from {corresponding} $ \supket{\gamma_S, \gamma_S} $ to specific sites. Due to Eq.~\eqref{eq: T random walk}, each step introduces a constant probability factor. Thus, the probability of reaching specific sites is $ \mathcal{O}(\frac{1}{n^{k/2}}) $ up to a log factor after logarithmic steps random walk. We put the proof details in Appendix \ref{appendix efficiency}.

\begin{figure}
    \centering
    \includegraphics[width=  \linewidth]{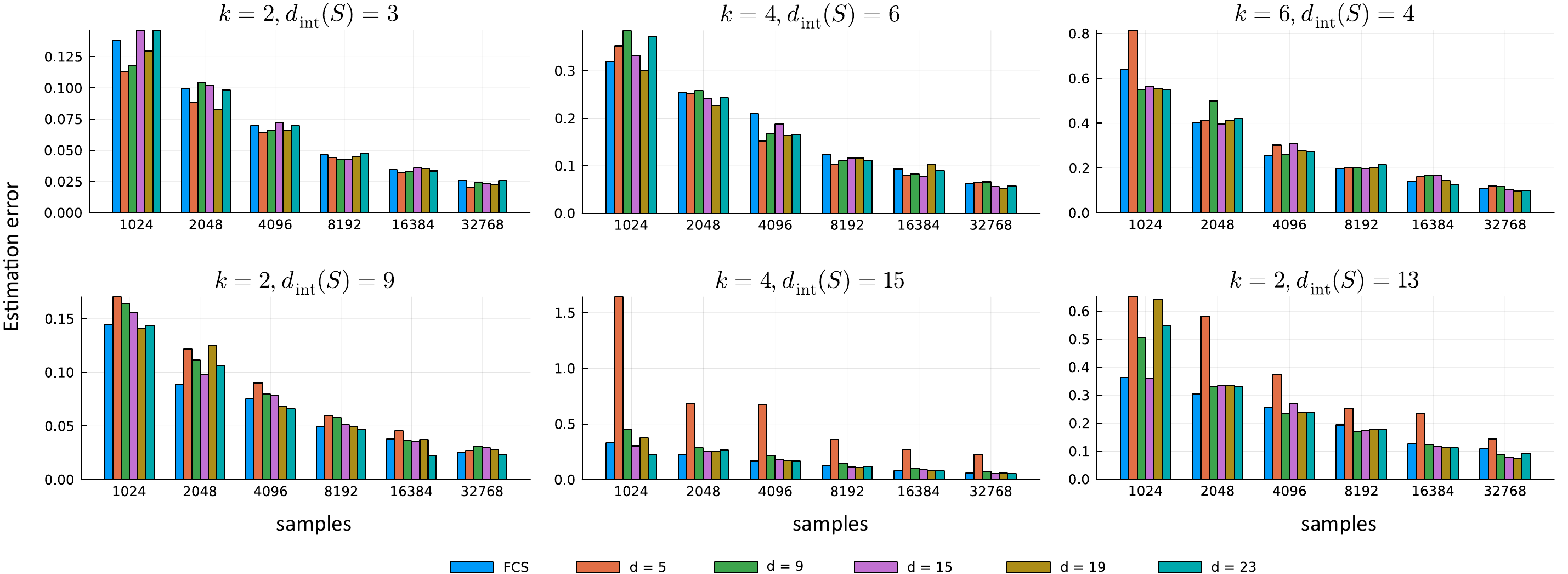}
    \caption{\centering
Estimation errors with increasing numbers of measurements for FCS and ADFCS protocols at different depths $d$ and interaction distance $d_\text{int}$. Each experiment uses a randomly generated $10$-qubit quantum state. 
 }
    \label{fig: figure2}
\end{figure}

We propose a general formula for approximating $ \alpha_{S,d} $. The proposed formula considers the random walk for computing  $ \alpha_{S,d}$  as the several independent $k=2$ random walks. 
Any pair $i,j\in S$  are treated as an independent random walk starting from {corresponding} $\gamma_i\gamma_j$. 
After the overall random walk, the probability of being at a specific site equals the product of the probabilities of each 2-local $ \gamma_{S'} $ at that site. As a result, the sum of probabilities at these specific sites can be expressed as a production $\prod\alpha_{\{i,j\},d}$. 
Therefore, we consider the overall random walk as the superposition of all such production
\begin{equation}
    \alpha_{S,d}' = \frac{1}{(k-1)!!} \left(\frac{3n}{2}\right)^\frac{k}{2} \frac{\binom{n}{k/2}}{\binom{2n}{k}}\sum_{\Lambda\in\mathrm{Par}(S)} \prod_{(i,j)\in \Lambda}\alpha_{\{i,j\},d}^L,
    \label{eq: fit alpha k}
\end{equation}
where $\mathrm{Par}(S)$ is the set whose elements $\Lambda$ are sets representing pairwise partitions of $S$. For example, let $S = \{1,2,3,4\}$, then $\Lambda = \{(1,2), (3,4)\}$ is a pairwise partitions of $S$.

Numerical experiments indicate that the proposed expression $ \alpha_{S,d}' $ closely approximates $ \alpha_{S,d} $ for depths $d > 2\log(n)$. When the interaction distance $d_\text{int}(S)=O(\log n)$, the optimal depth for random matchgate circuits is $O(\log n)$, making them shallow and well-suited for near-term quantum devices. Given this compatibility, our analysis focuses on bounding $\alpha_{S,d}$ with $ \alpha_{S,d}' $ in the regime where $d = \omega(\log n)$.
Therefore, it can be concluded that $ \alpha_{S,d}' $ provides a good approximation of $ \alpha_{S,d} $. Fig.~\ref{fig:sketch}(b) shows the results of comparison between $\alpha_{S,d}'$ and $\alpha_{S,d}$.
The curve of $\alpha_{S,d}'$ is very close to the $\alpha_{S,d}$ when $d> 2\log(n)$, which suggests that the $\alpha_{S,d}'$ is a good approximation of $\alpha_{S,d}$ as well.

We adaptively select the depth based on the interaction distance $ d_{\text{int}} $ of a fermionic Hamiltonian $H$, where $d_{\text{int}}(H)$ is the maximum value of $d_{\text{int}}(S)$ for terms $\gamma_S$ in $H$. Similar to the case with $|S|=2$, we adaptively select the depth $d^\ast$ as 
\begin{equation}
    d^\ast = \Theta\left(\max\left\{ \frac{d_{\text{int}}(S)^2}{\log(n)} , d_{\text{int}}(S)\right\}\right).
    \label{eq: adaptive depth}
\end{equation}
Such $ d^\ast $ ensures that {existing} a term in Eq.~\eqref{eq: fit alpha k} scales polynomially, thereby guaranteeing the overall $\alpha'_{S,d}$ also scales polynomially.

\section{Numerical experiment}

\begin{figure}
    \centering
    \includegraphics[width=\linewidth]{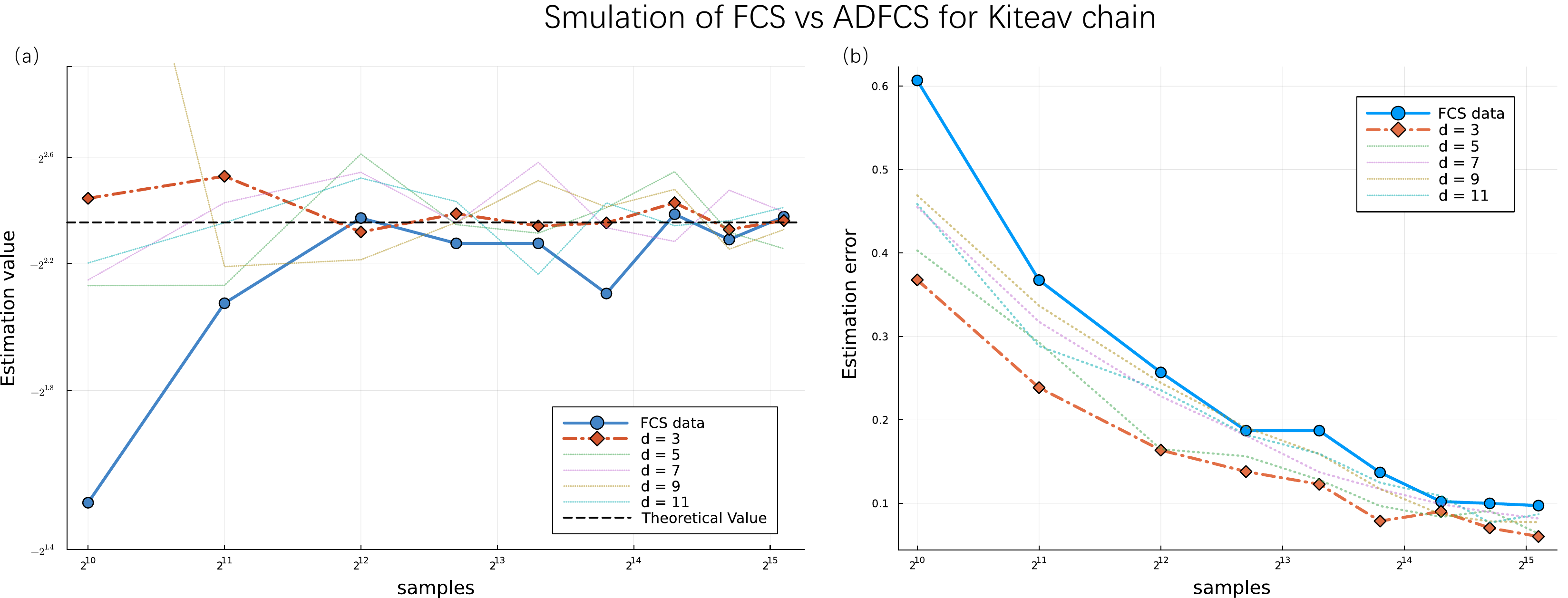}
    \caption{\centering
Application of ADFCS to the Kitaev chain Hamiltonian $H_K$. 
(a) The estimation generated from FCS and ADFCS with varying circuit depths.
(b) The error between the estimation and the expectation value. 
}
    \label{fig: figure3}
\end{figure}

Here, we provide numerical evidence demonstrating the efficiency of the ADFCS protocol. First, we compare the performance of ADFCS with FCS across various circuit depths and observables $\gamma_S$. For each experiment, a random $10$-qubit quantum state $\rho$ is chosen, followed by a $d$-depth random matchgate circuit with the brickwork structure and measurement on the computational basis. We enumerate depth $d$ in $\cbra{5,9,15,19,23}$. The estimation error with increasing numbers of measurements for ADFCS with different depth $d$ and FCS is shown in Fig.~\ref{fig: figure2}.
We estimate the error for the estimations of $\tra{\gamma_S \rho}$ for $k=\abs{S} \in\cbra{2,4,6}$ as the number of samples increases. Each subfigure corresponds to a different set $S$, chosen based on varying $d_\text{int}(S)$.
 The error is estimated as $\sqrt{\frac{1}{R}\sum_{i = 1}^R (\tra{\gamma_S\hat{\rho}_i} - \tra{\gamma_S\rho})^2} $, where $\hat{\rho}_i$ is generated from the $i$-th measurement, and $R$ is set to $64$.

The required depth of the random matchgate circuit for a specific observable $\gamma_S$ depends on the value of $\alpha_{S,d}$, thereby relying on the interaction distance $d_\text{int}(S)$.
Shallower circuits tend to perform better when the interaction distance is small, whereas larger interaction distances necessitate deeper circuits. For instance, when $d_\text{int}(S)\in \{3,4\} $, random matchgate circuits with depth $d = 5$ provide good estimations of the expectation value. However, for larger interaction distances, such as $d_\text{int}(S) = 13$ or $15$, a depth of $d = 5$ is insufficient to ensure accurate results. In such cases, circuits with a depth of $d \geq 9$ are necessary to achieve reliable estimations. By selecting an appropriate circuit depth, accurate estimations can be obtained with a relatively small number of samples, ensuring both efficiency and precision.
In each subfigure of Fig.~\ref{fig: figure2}, we use different random states for various choices of $S$. The numerical results demonstrate that ADFCS is robust to random states, despite the definition of $\alpha_{S,d}$ being based on the average over states $\rho$.

To further test the performance of ADFCS, we apply the method of estimating the expectation of the Kiteav chain Hamiltonian. The Kitaev chain Hamiltonian is a one-dimensional model that describes a topological superconductor featuring Majorana fermions at its ends~\cite{kitaev2001unpaired}. The Kiteav chain Hamiltonian is given by
\begin{equation}
    H_K = -\frac{\ii \mu}{2} \sum_{j=1}^n \gamma_{2j-1}\gamma_{2j} + \frac{\ii}{2}\sum_{j=1}^{n-1} \left( \omega_+ \gamma_{2j-1}\gamma_{2j+2} - \omega_- \gamma_{2j}\gamma_{2j+1} \right),
\end{equation}
where $\mu$ is the chemical potential, $\omega_{ \pm}=|\Delta| \pm t$, both $\Delta$ and $\Tcal$ are certain energy gaps. We initialize a random state for 10-qubit, set these parameters as $\mu = 2 $, $\Delta = 1$, and $t=0.4$. The results are shown in Fig.~\ref{fig: figure3}. Since the maximum interaction distance of $H_K$ is $d_\text{int}(H_K) = 3$, a random circuit with a depth of $d = 3$ yields an acceptable value of $\alpha_{S,d}$. This results in comparable or even improved estimation errors compared to the FCS protocol, while requiring significantly less circuit depth.

\section{Conclusion}
Estimating the expectation value of fermionic observables is a fundamental task in quantum physics and chemistry. The fermionic classical shadow method provides an innovative approach to address this problem without requiring the conversion of observables into the Pauli-basis representation. This significantly reduces the number of quantum states for certain local fermionic observables. However, the approach requires polynomial-depth quantum circuits when sampling matchgate elements from $\Mbb_n$ without considering the information from the fermionic observables.
This imposes a considerable challenge for near-term quantum devices.

We propose an adaptive depth fermionic classical shadow (ADFCS) protocol to reduce the heavy circuit depth associated with random matchgate circuits. Using a tensor network approach, we calculate the variance for any $d$-depth ADFCS protocol, enabling us to determine the required depth $d$ to ensure efficient sampling. 

Furthermore, we theoretically analyze the explicit relationship between the variance and the depth $d$ of the random matchgate circuit. Specifically, under certain assumptions, we find that the optimal depth is linear to $\max\left\{{d_{\text{int}}(H)^2}/{\log n}, d_{\text{int}}(H)\right\}$ where $d_{\text{int}}(H)$ is the interaction distance for the fermionic observable $H$. Numerical fitting results are also provided to support our theoretical findings.

We validate the correctness of our algorithm by evaluating the expectation values of several Majorana strings with respect to randomly generated quantum states $\rho$. The numerical results align with and support our theoretical findings. Additionally, we applied our algorithm to estimate the expectation value of a Kitaev chain Hamiltonian. Comparative numerical analysis demonstrates that our ADFCS algorithm achieves performance comparable to the FCS algorithm while requiring a significantly shallower circuit depth.

\section*{Acknowledgements}
 We would like to thank Dax Enshan Koh, Christian Bertoni, {and Jiapei Zhuang} for their helpful discussions.
This work was funded by the National Natural Science Foundation of China under Grant 12405014.

\bibliographystyle{unsrt}
% \bibliography{ref}

\appendix

\section{Introduction to Pauli-transfer matrix representation}
\label{appendix: intro_superoperator}

By employing the Jordan-Wigner transformation, we can express $\gamma_S$ in terms of Pauli operators. The Pauli-transfer matrix (PTM) representation uses these Pauli operators $\cbra{\sigma_i}_{i=1}^{2^n}$ as a basis, where $\sigma_i =P_i/\sqrt{2^n}$ is the normalized Pauli operator, allowing us to denote the non-zero linear operators $A$ as a $4^n$-length vector $\supket{A}$ with the $j$-th entry being
\begin{equation}
A_j = \text{Tr}(A \sigma_j).
\end{equation}
In this representation, a channel of operators can be expressed as a matrix $\Lambda$ with the $(a,b)$-th element being
\begin{equation}
\Lambda(a,b) = \text{Tr}(\sigma_b \Lambda(\sigma_a)).
\end{equation}
Using this expression, we can represent the second moment of the random matchgate as a fourth-order tensor, which is shown in Appendix~\ref{appendix 3}.

\section{Majorana operators diagonalize shadow channel}
\label{appendix: 1}

This section shows the details about diagonalizing shadow channel $\mathcal{M}_d$. 
Notice that the computational basis $\ket{b}$ are Gaussian states
\begin{equation}
    |b\rangle\langle b|=\prod_{j=1}^n \frac{1}{2}\left(I-\ii(-1)^{b_j} \gamma_{2 j-1} \gamma_{2 j}\right).
\end{equation}
Thus, any basis $\ketbra{b}{b}$ can be prepared from the state $\ketbra{b}{b}$ by a Gaussian unitary $U_Q\in \Mbb_n$, which indicates that Pauli-$X$ is in the matchgate group. 
The state $\ket{b}= \ket{b_1b_2\cdots b_n}$ can be denoted as $\prod X_{i}^{b_i} \ket{0}$. 
Since Pauli-$X$ is in the matchgate group, we can absorb the $\prod X_{b_i}$ into the matchgates $U_{Q_d}$ in the expression of shadow channel
\begin{align}
\mathcal{M}_d(\gamma_S) &= \int \dd U_{Q_d}\sum_{b\in\{0,1\}^n}\bra{b}U_{Q_d}\gamma_S U_{Q_d}^\dagger \ket{b} U_{Q_d}^\dagger \ket{b}\bra{b}U_{Q_d}\\
&= 2^n\int \dd U_{Q_d} \bra{0}U_{Q_d}\gamma_S U_{Q_d}^\dagger \ket{0} U_{Q_d}^\dagger \ket{0}\bra{0}U_{Q_d}.
\end{align}

If $S'$ is not equal to $ S$ and the depth $d$ is not equal to zero, 
then there exists a permutation matrix $Q_d'$ in one depth matchgate circuit such that
\begin{align}
\quad [\gamma_{S}, U_{Q_d'}] = 0, \quad \{\gamma_{S'}, U_{Q'_d}\} = 0. 
\end{align}
The orthogonal matrix $Q_d'$ can be constructed in one depth because such $U_{Q_d}'$ can be a Pauli operator, which involves one depth matchgate circuit. It implies
\begin{align}
\frac{1}{2^n}\tra{\gamma_{S'} \mathcal{M}_d(\gamma_S)} &= 
\int \dd U_{Q_d} \bra{0} U_{Q_d} \gamma_S U_{Q_d}^\dagger \ket{0}\bra{0}U_Q \gamma_{S'} U_{Q_d}^\dagger\ket{0}\\
&=\int \dd U_{Q_d} \bra{0}U_{Q_d} U_{Q'_d} \gamma_S U_{Q'_d}^\dagger U_{Q_d}^\dagger\ket{0}\bra{0}U_{Q_d} U_{Q'_d}\gamma_{S'} U_{Q'_d}^\dagger U_{Q_d}^\dagger\ket{0}\\
&= -\int \dd U_{Q_d} \bra{0}U_{Q_d} \gamma_S U_{Q_d}^\dagger\ket{0}\bra{0}U_{Q_d} \gamma_{S'} U_{Q_d}^\dagger\ket{0}.
\end{align}
The result shows that $\tra{\gamma_{S'} \mathcal{M}_d(\gamma_S)} = 0$ when $S'$ is not equal to $ S$, thereby $\mathcal{M}_d (\gamma_S) = \alpha_{s,d} \gamma_S$.

\section{Bound the variance of ADFCS estimator with $\alpha_{S,d}$}
\label{appendix: bound variance}

The variance of the random variable $ v = \operatorname{Tr}(\hat{\rho} \gamma_S) $ is analyzed by deriving an upper bound on its value. First, the variance is bounded by the expected squared magnitude of $ v $, i.e., $ \Var[v] \leq \mathbb{E}[|v|^2] $. The expectation is then expressed as an integral over the unitary group $ U_{Q_d} $ and averaged over the state $ \rho $
\begin{equation}
    \mathbb{E}[|v|^2] = \int \mathrm{d}U_{Q_d} \, \mathbb{E}_{\rho}\left[\sum_b \langle b|U_{Q_d} \rho U_{Q_d}^\dagger|b \rangle \left|\langle b|U_{Q_d} \mathcal{M}_d^{-1}(\gamma_S) U_{Q_d}^\dagger|b \rangle\right|^2\right].
\end{equation}

Now, we start to simplify the expression.
Firstly, we average out the $\rho$, which leads to
\begin{equation}
    \mathbb{E}[|v|^2] = 2^{-n} \int \mathrm{d}U_{Q_d} \sum_b \left|\langle b|U_{Q_d} \mathcal{M}_d^{-1}(\gamma_S) U_{Q_d}^\dagger|b \rangle\right|^2.
\end{equation}

Next, we denote $\ket{b}$ as $\prod_i X_i^{b_i}\ket{0}$ and absorb the Pauli X operators into the matchgate operators.                                          
And then uses Eq.~\eqref{eq: lemma1 eigen}
\begin{align}
    \mathbb{E}[|v|^2] =& \frac{1}{|\alpha_{S,d}|^2} \int \langle 0|U_{Q_d} \gamma_S U_{Q_d}^\dagger|0 \rangle \langle 0|U_{Q_d} \gamma_S^\dagger U_{Q_d}^\dagger|0 \rangle \\
 &=\frac{1}{\alpha_{S,d}}.
\end{align}
Finally we have $\Var[v] \leq \frac{1}{\alpha_{S,d}}$. 
Notice that this inequality can be easily taken equally when $\tra{\rho \gamma_S} = 0$.

\section{Details of simplifying $\alpha_{S,d}$ to tensore network}
\label{appendix 3}
 This section shows the details of representing $\alpha_{S,d}$ by the tensor network contraction in PTM representation. According to Lemma \ref{lemma: The net phase of alpha is 1}, the $\alpha_{S,d}$ can be expressed by 
 \begin{equation}
     \alpha_{S,d} = 2^{2n}  \supbra{\bm 0,\bm 0}  \int\dd U_{Q_d} \mathcal{U}_{Q_d}^{\otimes 2} \supket{P_S,P_S}, 
 \end{equation}
where $P_S$ is the Pauli string corresponding to $\gamma_S$ via Jordan-Wigner transformation. 

Since each two-qubit random matchgate is independently sampled, the integral $\int \dd U_Q \mathcal{U}_Q^{\otimes 2}$ can be calculated by independently integrating each 2-qubit matchgate. The result of the integral of the $2$-qubit matchgates is given by Lemma \ref{lemma: 1},
\begin{equation}
\label{eq: integral of the 2 qubits matchgates}
\begin{aligned}
    \int_{U_Q\sim \Mbb_2} \dd U_Q\mathcal{U}_Q^{\otimes 2} =& \supketbra{\gamma_\emptyset}{\gamma_\emptyset}
    + \frac{1}{4} \sum_{i,j} \supketbra{\gamma_i}{\gamma_j}\\
    &+ \frac{1}{6}\sum_{\substack{i_1\neq i_2 \\ j_1\neq j_2}}\supketbra{\gamma_{i_1}\gamma_{i_2}}{\gamma_{j_1}\gamma_{j_2}} \\
    &+ \frac{1}{4}
    \sum_{\substack{i_1\neq i_2, j_1 \neq j_2 \\ 
        i_1\neq i_3, j_1 \neq j_3 \\
        i_2\neq i_3, j_2 \neq j_3} 
    }
    \supketbra{\gamma_{i_1}\gamma_{i_2}\gamma_{i_3}}{\gamma_{j_1}\gamma_{j_2}\gamma_{j_3}}\\
    &+ \supketbra{\gamma_1\gamma_2\gamma_3\gamma_4}{\gamma_1\gamma_2\gamma_3\gamma_4},
\end{aligned}
\end{equation}
where $i$, $j$ are index ranged from $1$ to $4$. Following the definition Eq.~\eqref{eq: T}, each element of tensor $\Tcal$ can be calculated by Eq.~\eqref{eq: integral of the 2 qubits matchgates}.
We show these concrete elements in Table \ref{table: the 16x16 matrix of T}.  The tensor $\Tcal$ presents the average effect of a random two-qubit matchgate.

Here, we represented $\supket{\bm 0,\bm 0}$ in PTM to complete the calculation $\alpha_{S,d} = \supbra{\bm 0,\bm 0} \mathcal{C}\supket{\gamma_S, \gamma_S}$.
Notice the matrix identity 
\begin{equation}
    \ketbra{\bm 0} = \frac{1}{2^n} \sum_{\Lambda\subseteq [n]} \prod_{i \in \Lambda} Z_i,
\end{equation}
where $Z_i$ denotes the application of the Pauli $Z$ operator to the $i$-th qubit.
Especially, when $\Lambda = \emptyset$, let $\prod_{i \in \Lambda} Z_i = \mathbbm{1}_n$. Then, the super vector of $\supket{\bm 0, \bm 0}$ can be expressed as
\begin{equation}
\label{eq: zz anonymous 6}
    \supket{\bm 0, \bm 0} = \frac{1}{2^{2n}} \sum_{\Lambda, \Lambda' \subseteq [n]}  \supket{\prod_{i \in \Lambda }  Z_i, \prod_{j \in \Lambda'} Z_j}.
\end{equation}
Eq.~\eqref{eq: zz anonymous 6} express the supervector $\supket{\bm 0, \bm 0}$ in the PTM representation. Finally, we write the $\mathcal{C}$, $\supket{\bm 0, \bm 0}$, and $\supket{\gamma_S, \gamma_S}$ in Pauli basis.
Thus, the $\alpha_{S,d}$ calculation can be expressed as the tensor network contraction, as illustrated in Fig.~\ref{fig: figure4}.

\begin{table}
\centering
\caption{Values of tensor $\Tcal$. The head of columns represents the input of $\Tcal$ while the head of rows represents the output of $\Tcal$. For example, the values in row `XY' and column `YZ' represent the value $\Tcal^{XY}_{~~YZ}$.  The blank space of the table stands for $0$. For example,  $\Tcal^{XY}_{~~YZ}$ is equal to $\frac{1}{6}$.   }
\begin{tabular}{|c|c|c|c|c|c|c|c|c|c|c|c|c|c|c|c|c|}
	\hline
	$\Tcal$  &II &IX &IY &IZ &XI &XX &XY &XZ &YI &YX &YY &YZ &ZI &ZX &ZY &ZZ \\ \hline
    II & 1 &   &   &   &   &   &   &   &   &   &   &   &   &   &   &   \\ \hline
    IX &   &1/4&1/4&   &   &   &   &1/4&   &   &   &1/4&   &   &   &   \\ \hline
    IY &   &1/4&1/4&   &   &   &   &1/4&   &   &   &1/4&   &   &   &   \\ \hline
    IZ &   &   &   &1/6&   &1/6&1/6&   &   &1/6&1/6&   &1/6&   &   &   \\ \hline
    XI &   &   &   &   &1/4&   &   &   &1/4&   &   &   &   &1/4&1/4&   \\ \hline
    XX &   &   &   &1/6&   &1/6&1/6&   &   &1/6&1/6&   &1/6&   &   &   \\ \hline
    XY &   &   &   &1/6&   &1/6&1/6&   &   &1/6&1/6&   &1/6&   &   &   \\ \hline
    XZ &   &1/4&1/4&   &   &   &   &1/4&   &   &   &1/4&   &   &   &   \\ \hline
    YI &   &   &   &   &1/4&   &   &   &1/4&   &   &   &   &1/4&1/4&   \\ \hline
    YX &   &   &   &1/6&   &1/6&1/6&   &   &1/6&1/6&   &1/6&   &   &   \\ \hline
    YY &   &   &   &1/6&   &1/6&1/6&   &   &1/6&1/6&   &1/6&   &   &   \\ \hline
    YZ &   &1/4&1/4&   &   &   &   &1/4&   &   &   &1/4&   &   &   &   \\ \hline
    ZI &   &   &   &1/6&   &1/6&1/6&   &   &1/6&1/6&   &1/6&   &   &   \\ \hline
    ZX &   &   &   &   &1/4&   &   &   &1/4&   &   &   &   &1/4&1/4&   \\ \hline
    ZY &   &   &   &   &1/4&   &   &   &1/4&   &   &   &   &1/4&1/4&   \\ \hline
    ZZ &   &   &   &   &   &   &   &   &   &   &   &   &   &   &   & 1 \\ \hline
\end{tabular}
\label{table: the 16x16 matrix of T}
\end{table}

\begin{lemma}
\label{lemma: The net phase of alpha is 1}
The quantity $\alpha_{S,d}$ satisfies the identity $\alpha_{S,d} = 2^{2n} \supbra{\bm{0}, \bm{0}} \int \dd U_{Q_d} \mathcal{U}_{Q_d}^{\otimes 2} \supket{P_S, P_S}$, where $P_S$ denotes the Pauli string associated with $\gamma_S$ through the Jordan-Wigner transformation.
\end{lemma}
\begin{proof}

Follow the definition of $\alpha_{S,d}$, we have
\begin{align}
    \alpha_{S,d} =& \int_{Q\sim O_d } \dd U_{Q_d} \abs{\bra{\bm 0} U_{Q_d} \gamma_S U_{Q_d}^\dagger \ket{\bm 0}}^2 \\
    =& \int_{Q\sim O_d } \dd U_{Q_d} \bra{\bm 0} U_{Q_d} \gamma_S U_{Q_d}^\dagger \ket{\bm 0}\bra{\bm 0} U_{Q_d} \gamma_S^\dagger U_{Q_d}^\dagger \ket{\bm 0}.
\end{align}
The subscript $d$ denotes the depth of the matchgate circuit. 
The expression can be simplified by substituting the relationship between $\gamma_S^\dagger$ and $\gamma_S$, which is 
\begin{equation}
    \gamma_S^\dagger = (-1)^{\frac{|S|(|S|-1)}{2}}\gamma_S,
\end{equation}
due to the anti-commutation relation of Majorana operators $\{\gamma_i, \gamma_j\} = 2\delta_{ij}$. 
The $\alpha_{S, d}$ can be expressed as  
\begin{align}
    \alpha_{S,d} =& (-1)^{\frac{|S|(|S|-1)}{2}}\int \dd U_{Q_d} 
    \bra{\bm 0}U_{Q_d} \gamma_SU_{Q_d}^\dagger\ket{\bm 0}^2\\
    =& (-1)^{\frac{|S|(|S|-1)}{2}} \int \dd U_{Q_d} 
    \tr(U_{Q_d} \gamma_SU_{Q_d}^\dagger\ket{\bm 0} \bra{\bm 0} )^2\\
    =& (-1)^{\frac{|S|(|S|-1)}{2}} 2^{2n} \int \dd U_{Q_d} 
    \supbra{\bm 0,\bm 0} \mathcal{U}_{Q_d}^{\otimes 2} \supket{\gamma_S,\gamma_S} \\
    =& (-1)^{\frac{|S|(|S|-1)}{2}} 2^{2n}  \supbra{\bm 0,\bm 0}  \int\dd U_{Q_d} \mathcal{U}_{Q_d}^{\otimes 2} \supket{\gamma_S,\gamma_S}.
    \label{appendix eq: alpha to net}
\end{align}

In PTM representation, the $\gamma_S$ corresponds to a Pauli basis with a phase $\pm \ii^{\floor{|S|/2}}$. Thus, the super vector $\supket{\gamma_S, \gamma_S}$ can be represented as 
\begin{equation}
    \supket{\gamma_S, \gamma_S} = (-1)^{\floor{\frac{|S|}{2}}} \supket{P_S, P_S},
    \label{appendix eq: gamma to P}
\end{equation}
 where $P_S$ is the Pauli operator that corresponds to the $\gamma_S$.

    Follows Eq.~\eqref{appendix eq: alpha to net} and Eq.~\eqref{appendix eq: gamma to P}, 
    the sign of $\alpha_{S,d}$ is determined by $(-1)^{\frac{|S|(|S|-1)}{2}} (-1)^{\floor{|S|/2}}$. 
    We show that the sign equals to $1$ by categorizing the parity of $|S|$. 
    \begin{enumerate}
        \item \textbf{$|S|$ is an odd number.} Let $|S| = 2q+1$, $ q\in \mathbb{N}$, $q\geq 0$. And then
        \begin{equation}
            (-1)^{\frac{|S|(|S|-1)}{2}} (-1)^{\floor{|S|/2}} = (-1)^{q(2q+1)+q} = 1.
        \end{equation}
        \item \textbf{$|S|$ is an even number.} Let $|S| = 2q$, $ q\in \mathbb{N}$, $q\geq 0$. And then
        \begin{equation}
            (-1)^{\frac{|S|(|S|-1)}{2}} (-1)^{\floor{|S|/2}} = (-1)^{(2q-1)q+q} = 1.
        \end{equation}
    \end{enumerate}
Thus, we conclude that the $\alpha_{S,d}$ can be expressed by 
\begin{align}
    \alpha_{S,d}=& (-1)^{\frac{|S|(|S|-1)}{2}} (-1)^{\floor{|S|/2}}  2^{2n}  \supbra{\bm 0,\bm 0}  \int\dd U_{Q_d} \mathcal{U}_{Q_d}^{\otimes 2} \supket{P_S,P_S}\\
    &= \supbra{\bm 0,\bm 0}  \int\dd U_{Q_d} \mathcal{U}_{Q_d}^{\otimes 2} \supket{P_S,P_S} .
\end{align}
\end{proof}

\section{Mapping the action of tensors to random walk}
\label{sec: Mapping the action of tensors to random walk}

Here, we give some notations to express the properties of the $\alpha_{S,d}$ tensor network. Let $\Tcal^{(o)}$ represent the odd layer in $\mathcal{C}$, and $\Tcal^{(e)}$ represent the even layer in $\mathcal{C}$, as shown in Fig.~\ref{fig:Illustration_odd_even_layers}. 
Any tensor network $\mathcal{C}$ with brickwork architecture can be expressed by alternately apply $\Tcal^{(o)}$ and $\Tcal^{(e)}$ gates, 
\begin{equation}
\mathcal{C} \left(t, b_1, b_2\right)=\Tcal^{(o)^{b_2}}\mathcal{B}^t \Tcal^{(e)^{b_1}}, 
\end{equation}
where $b_1, b_2 \in\{0,1\}$, $t+b_1+b_2$ stands for the depth, and $\mathcal{B = }\Tcal^{(e)}\Tcal^{(o)}$. 
Let $\Gamma'_n$ denote the vectorized double supervector space of $\Gamma_n$, defined as $\Gamma'_n = \text{span}\{\supket{\gamma_S, \gamma_S} \mid \gamma_S \in \Gamma_n \}$. Notably, both $\Tcal^{(o)}$ and $\Tcal^{(e)}$ gates are projection operators, satisfying $(\Tcal^{(e)})^2 = \Tcal^{(e)}$ and $(\Tcal^{(o)})^2 = \Tcal^{(o)}$.

\begin{figure}[H]
    \centering
    \includegraphics[width=0.22\linewidth]{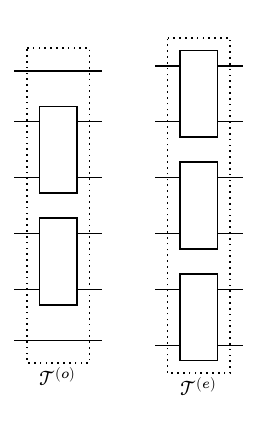}
    \caption{\centering Illustration for the odd-layer tensor $\Tcal^{(o)}$ and the even-layer tensor $\Tcal^{(e)}$. }
\label{fig:Illustration_odd_even_layers}
\end{figure}
Throughout the main text, we often use the term “representation” without giving a precise definition. In this section, we use language from representation theory to provide a clear and more formal explanation of how \(\mathcal{C}\) is represented. 
As we mentioned in the main text, the whole tensor $\mathcal{C}$ can be expressed by
\begin{equation}
    \mathcal{C} =\begin{cases}
        \mathcal{B}^t, ~~d = 2t\\
        \mathcal{B}^t \Tcal_\text{init}, ~~d = 2t+1.
    \end{cases}
\end{equation}
For the odd-depth case, studying the contraction of \(\mathcal{C}\) on \(\Gamma'_n\) is the same as analyzing how the group \(\{\mathcal{B}^t\}\) acts on the representation space \(\Tcal_\text{init}(\Gamma'_n)\). A similar approach holds for even depths with \(t \geq 1\), except the initial states differ: for even depth, the initial state is \(\Tcal_\text{init}\supket{\gamma_S, \gamma_S}\), while for odd depth, it is \( \mathcal{B} \supket{\gamma_S, \gamma_S}\). Since \( \mathcal{B} \supket{\gamma_S, \gamma_S}\) still lies in \(\Tcal_\text{init}(\Gamma'_n)\), the same method applies to the even-depth case. Therefore, whether the depth is even or odd, the contraction of \(\mathcal{C}\) on \(\Gamma'_n\) can be studied by examining the representation of the group \(\{\mathcal{B}^t\}\). We refer to this simply as “the representation of \(\mathcal{C}\)”. To simplify our analysis, we focus on a random circuit with an even number of qubits and an odd depth \(d\). We will represent the contraction of \(\mathcal{C}\) in a polynomial space, as shown in Fig.~\ref{diagram: commute diagram}.

\subsection{Reduce the calculation to polynomial space}
\begin{figure}
    \centering
    \includegraphics[width=0.6\linewidth]{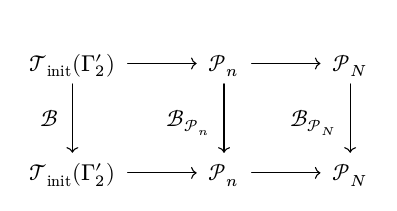}
    \caption{\centering Illustration for the isometric mapping to the polynomial space. $\Gamma_2'$ represents the vectorized double supervector space. Horizontal arrows indicate transitions between spaces, while vertical arrows correspond to the operators $\mathcal{B}$ in different representations. Ultimately, the process is reduced to the $N$-elementary polynomial within the second-degree polynomial space.}
    \label{diagram: commute diagram}
\end{figure}

We represent \(\mathcal{C} \supket{\gamma_S, \gamma_S}\) in a polynomial space to allow key operations, such as multiplication and addition, that are essential for our analysis. According to Table~\ref{table: the 16x16 matrix of T}, \(\Tcal_\text{init}(\Gamma'_2)\) is spanned by operators of the form
\begin{equation}
    \begin{cases}
    \supket{\psi_{ij}} =& \frac{1}{4}\supket{X_i \left(\displaystyle\prod_{k = i+1}^{j-1} Z_k\right)Y_j}^{\otimes 2}
    + \frac{1}{4}\supket{X_i \left(\displaystyle\prod_{k = i+1}^{j-1} Z_k\right)Y_j}^{\otimes 2} \\
    & + \frac{1}{4}\supket{X_i \left(\displaystyle\prod_{k = i+1}^{j-1} Z_k\right)Y_j}^{\otimes 2}
    + \frac{1}{4}\supket{X_i \left(\displaystyle\prod_{k = i+1}^{j-1} Z_k\right)Y_j}^{\otimes 2}, \quad i<j, \\[6pt]
    \supket{\psi_{ii}} =& \supket{Z_i, Z_i}.
\end{cases}
\end{equation}

We define an isometric map \(\phi: \Tcal_\text{init}(\Gamma'_2) \to \mathcal{P}_n\) by setting \(\phi(\supket{\psi_{ij}}) := x_i x_j\) and \(\phi(\supket{\psi_{ii}}) := x_i^2\). This map is a linear isomorphism and a homomorphism with an inverse, which means \(\Tcal_\text{init}(\Gamma'_2)\) and \(\mathcal{P}_n\) are isomorphic. We then define the action of \(\mathcal{C}\) on \(\mathcal{P}_n\) by
\begin{equation}
    \mathcal{C}_{\mathcal{P}_n}(\cdot) := \phi \,\circ\, \mathcal{C} \,\circ\, \phi^{-1} (\cdot).
\end{equation}
Consequently, the action of a \(d\)-depth matchgate circuit on \(\gamma_{\{i,j\}}\) is equivalent to the action of \(\mathcal{C}_{\mathcal{P}_n}\) on \(x_i x_j\).

By Lemma~\ref{lemma: reduce Pn to PN}, we can further simplify this representation by identifying a sub-representation \(\mathcal{P}_N \subset \mathcal{P}_n\), where \(N = \tfrac{n}{2}\). The variable \(N\) is an integer since we have assumed that $n$ is an even number. Certain hidden patterns are revealed by reducing the representation to the polynomial space $\mathcal{P}_N$. Most of the main results are been proved in the $\mathcal{P}_N$.

\begin{lemma}
    \label{lemma: reduce Pn to PN}
    The space $\mathcal{P}_N$ is a faithful representation of $\{\mathcal{B}\}$ restricts on the space $\Tcal_\text{init}(\Gamma_2')$
\end{lemma}
\begin{proof}
We begin by defining the group action on \(\mathcal{P}_N\) by using \(\mathcal{P}_n\) as a medium. Since we have defined the group action on \(\mathcal{P}_n\), we want to define the group action on \(\mathcal{P}_N\) with the help of \(\mathcal{P}_n\). With this idea, we define a linear map $\varphi$ as
\begin{equation}
    \begin{aligned}
        \varphi : \mathcal{P}_N &\to \mathcal{P}_n \\
          y_i^2 &\mapsto x_{2i-1}^2+4 x_{2i-1} x_{2i}+x_{2i}^2\\
       y_i y_j &\mapsto  \left(x_{2i-1}+x_{2i}\right)\left(x_{2j-1}+x_{2j}\right).
    \end{aligned}
\end{equation}
    
For a general element in \(\mathcal{P}_N\), written as \(\sum \xi_{ij}\,y_i y_j\), the linearity leads to
\[
\varphi\Bigl(\sum \xi_{ij}\,y_i y_j\Bigr)
\;:=\;
\sum \xi_{ij}\,\varphi(y_i y_j).
\]
    This makes \(\varphi(\mathcal{P}_N)\) a linear subspace of \(\mathcal{P}_n\). Because \(\varphi\) is injective, there is a map \(\varphi': \mathcal{P}_n \to \mathcal{P}_N\) such that \(\varphi' \circ \varphi\) is the identity on \(\mathcal{P}_N\).  
    
    The action of $\mathcal{C}$ in $\mathcal{P}_N$ is constructed by 
    \begin{equation}
        \mathcal{C}_{\mathcal{P}_N}(y_i y_j) := \varphi' \circ \mathcal{C}_{\mathcal{P}_n}\circ \varphi(y_i y_j).
        \label{eq: appendix T whole PN}
    \end{equation}
    Based on the definition of representation, the space $\mathcal{P}_N$ is a representation of $\mathcal{C}$. And Eq.~\eqref{eq: appendix T whole PN} depicts how to act a tensor on $\mathcal{P}_N$.
    
    The next step is to show that such a representation of $\mathcal{C}$ does not ``lose information". Formally, we need to prove that the representation $\mathcal{P}_n$ can be reduced to a sub-representation isometrics to $\mathcal{P}_N$. Naturally, we will consider whether the subspace $\varphi(\mathcal{P}_N) $ will constitute a sub-representation. If it is true, we can reduce the representation to its faithful sub-representation $\varphi(\mathcal{P}_N)\subseteq{P}_n $, which isometric to the polynomial space $\varphi(\mathcal{P}_N) \simeq \mathcal{P}_N$.

    By the definition of sub-representation, we need to prove
    \begin{equation}
    \label{eq: PN is a sub representation}
         \mathcal{B}^t (\varphi(\mathcal{P}_N)) \subseteq \varphi(\mathcal{P}_N), ~~\forall t .
    \end{equation}

     The proof of statement \eqref{eq: PN is a sub representation} will be carried out using mathematical induction.
    
     When $t=0$, $\varphi(\mathcal{P}_N)\subseteq \varphi(\mathcal{P}_N)$ is true.
     Suppose the statement is true for $t^*$, then we have
    \begin{equation}
        (\Tcal^{(e)}_{\mathcal{P}_n}\Tcal^{(o)}_{\mathcal{P}_n})^{t^*}(\varphi(y_iy_j)) = \sum \xi_{lm} \varphi(y_ly_m).
    \end{equation}
    And then
    \begin{align}
        \mathcal{B}_{\mathcal{P}_n}  ^{t^*+1}(\varphi(y_i y_j)) =& 
        \Tcal^{(e)}_{\mathcal{P}_n}\Tcal^{(o)}_{\mathcal{P}_n} (\sum \xi_{lm} \varphi(y_ly_m)) \\
        \label{eq: zz anonymous 2}
        =& \sum \xi_{lm} \delta_{lm}\Tcal^{(e)}_{\mathcal{P}_n}\Tcal^{(o)}_{\mathcal{P}_n} ( x_{2i-1}^2+4 x_{2i-1} x_{2i}+x_{2i}^2) \\
        \label{eq: zz anonymous 3}
        & + \sum \xi_{lm} (1-\delta_{lm})\Tcal^{(e)}_{\mathcal{P}_n}\Tcal^{(o)}_{\mathcal{P}_n}\left(\left(x_{2 i-1}+x_{2 i}\right)\left(x_{2 j-1}+x_{2 j}\right)\right),
    \end{align}
    where $\mathcal{B}_{\mathcal{P}_n}$,  $\Tcal^{(e)}_{\mathcal{P}_n}$ and $\Tcal^{(o)}_{\mathcal{P}_n}$ are defined in a manner analogous to $\mathcal{C}_{\mathcal{P}_n}$. 
    We will prove Eq.~\eqref{eq: zz anonymous 2} is in the space $\varphi(\mathcal{P}_N)$, while the same result of Eq.~\eqref{eq: zz anonymous 3} can be proved by straightforward calculation
    \begin{equation}
    \begin{aligned}
        \label{eq: zz anonymous 4}
        &\Tcal^{(e)}_{\mathcal{P}_n}\Tcal^{(o)}_{\mathcal{P}_n} ( x_{2i-1}^2+4 x_{2i-1} x_{2i}+x_{2i}^2) \\
        =& \Tcal^{(e)}_{\mathcal{P}_n}(\frac{1}{6}x_{2i-2}^2 + \frac{2}{3} x_{2i-2}x_{2i-1} + x_{2i-2}x_{2i}+x_{2i-2}x_{2i+1} \\
        &+\frac{1}{6} x_{2i-1}^2 + x_{2i-1}x_{2i} + x_{2i-1}x_{2i+1} +\frac{1}{6} x_{2i}^2 + \frac{2}{3} x_{2i}x_{2i+1} + \frac{1}{6}x_{2i+1}^2  )\\
        =& \frac{1}{36}x_{2i-3}^2 + \frac{1}{9} x_{2i-3}x_{2i-2} +\frac{5}{12} x_{2i-3}x_{2i-1}+ \frac{5}{12} x_{2i-3}x_{2i}  + \frac{1}{4} x_{2i-3}x_{2i+1} + \frac{1}{4} x_{2i-3}x_{2i+2}  \\
        &+ \frac{1}{36}x_{2i-2}^2 + \frac{5}{12} x_{2i-2}x_{2i-1}+ \frac{5}{12} x_{2i-2}x_{2i}  + \frac{1}{4} x_{2i-2}x_{2i+1} + \frac{1}{4} x_{2i-2}x_{2i+2}  \\
        &+ \frac{2}{9}x_{2i-1}^2 + \frac{8}{9} x_{2i-1}x_{2i} + \frac{5}{12} x_{2i-1}x_{2i+1} + \frac{5}{12} x_{2i-1}x_{2i+2}\\
        & + \frac{2}{9}x_{2i}^2 + \frac{5}{12} x_{2i}x_{2i+1} + \frac{5}{12} x_{2i}x_{2i+2} + \frac{1}{36}x_{2i+1}^2  + \frac{1}{9} x_{2i+1}x_{2i+2}+ \frac{1}{36}x_{2i+2}^2 \\
    \end{aligned}
    \end{equation}
    The result expression is in $\varphi(\mathcal{P}_N)$ because we can find a polynomial $y$ in $\mathcal{P}_N$ such that $\varphi(y)$ equals the result expression,
    \begin{align}
        y =& \frac{1}{6} y_{i-1}^2 + \frac{5}{3} y_{i-1}y_{i} + y_{i-1}y_{i+1} + \frac{4}{3} y_{i}^2 + \frac{5}{3} y_{i-1}y_{i+1}  + \frac{1}{6} y_{i+1}^2 \\
        \varphi(y) =& \Tcal^{(e)}_{\mathcal{P}_n}\Tcal^{(o)}_{\mathcal{P}_n} ( x_{2i-1}^2+4 x_{2i-1} x_{2i}+x_{2i}^2)
    \end{align}
    Thus, we get that 
    \begin{equation}
        \mathcal{B}_{\mathcal{P}_n} ( x_{2i-1}^2+4 x_{2i-1} x_{2i}+x_{2i}^2) \subseteq \varphi(\mathcal{P}_N),
    \end{equation}
    which complete the proof of statement \ref{eq: PN is a sub representation}.

We aim to show that $\varphi(\mathcal{P}_N)$ is a faithful sub-representation. Suppose, for the sake of contradiction, that $\varphi(\mathcal{P}_N)$ is not faithful. Then, there exists a non-zero element $y \in \mathcal{P}_N$ such that:
\[
\varphi(y) \neq 0 \quad \text{and} \quad \mathcal{B}_{\mathcal{P}_n}^t (\varphi(y)) = 0 \quad \text{for some integer } t \geq 1.
\]
Consider applying the operator $\mathcal{B}_{\mathcal{P}_n}^{~t'}$ iteratively to both sides of the equation:
\[
\mathcal{B}_{\mathcal{P}_n}^{~t'} \left[ \mathcal{B}_{\mathcal{P}_n}^{~t} (\varphi(y)) \right] = \mathcal{B}_{\mathcal{P}_n}^{t + t'} (\varphi(y)) = \mathcal{B}_{\mathcal{P}_n}^{~t'} (0) = 0.
\]
Taking the limit as $t' \to \infty$, we obtain:
\[
\lim_{t' \to \infty} \mathcal{B}_{\mathcal{P}_n}^{t + t'} (\varphi(y)) = 0.
\]
However, the infinite application of $\mathcal{B}_{\mathcal{P}_n}$ to $\varphi(y)$ yields a non-zero polynomial.
This is a contradiction because the same expression cannot simultaneously be zero and non-zero. Therefore, our initial assumption that $\varphi(\mathcal{P}_N)$ is not faithful must be false. 
Hence, $\varphi(\mathcal{P}_N)$ is indeed a faithful sub-representation.

\end{proof}

\subsection{Mapping the spread of polynomials to random walk}
\label{sec: mapping the action of tensors to random walk}

\newcommand{\teto}{\mathcal{B}_{\mathcal{P}_N}}
We can calculate the output of $\teto (y_iy_j)$ by $\teto(y_i y_j) := \varphi' \circ \phi^{-1} \circ \mathcal{B}\circ \phi \circ \varphi(y_i y_j)$, and the results are shown in Table \ref{table: transition table of teto}. The definition of function $L$ is
\begin{equation}
\label{eq: lazy random walk L}
L\left(y_i\right)=
\begin{cases}
\frac{3}{4} y_1+\frac{1}{4} y_2, &i=1 \\
\frac{1}{4} y_{i-1}+\frac{1}{2} y_i+\frac{1}{4} y_{i+1}, &1<i<N \\
\frac{3}{4} y_N+\frac{1}{4} y_{N-1}, &i=N. 
\end{cases}
\end{equation}
Notice that for most cases, the subscripts $i,j$ satisfy $1<i<N-1$ and $i+1<j\leq N$, then the outcomes of $\teto$ is
\begin{equation}
    \label{eq: zz anonymous 5}
    \teto (y_iy_j) = \left(\frac{1}{4} y_{i-1} + \frac{1}{2}y_i + \frac{1}{4} y_{i-1}\right)\left( \frac{1}{4} y_{j-1} + \frac{1}{2}y_j + \frac{1}{4} y_{j-1} \right).
\end{equation}
In this case, the action of $\teto$ can be viewed as independently evolving $y_i$ and $y_j$ in a one-dimensional lattice (see Fig.~\ref{fig: lazy random walk})
\begin{equation}
\label{eq: separate evolution}
    \begin{aligned}
        y_i &\to \frac{1}{4} y_{i-1} + \frac{1}{2}y_i + \frac{1}{4} y_{i-1}\\
    y_j &\to \frac{1}{4} y_{j-1} + \frac{1}{2}y_j + \frac{1}{4} y_{j-1}.
    \end{aligned}
\end{equation}
Here, we interpreted $y_1, \cdots y_n$ as a finite one-dimensional lattice whose $i$-th site is labeled by $y_i$. The action of $\mathcal{B}_{\mathcal{P}_N}$ can be interpreted as a random walk on 2D lattice whose $(i,j)$ site is labeled by $y_iy_j$. 
We observe that the equation $\mathcal{B}_{\mathcal{P}_N}(y_iy_j) = L(y_i)L(y_j)$ holds for most sites. This observation inspires us to first analyze the evolving behavior of $L$ and then use the probability distribution of $L$ to estimate the order of $\alpha_{S,d}$.

\begin{figure}
    \centering
    \includegraphics[width=0.9\linewidth]{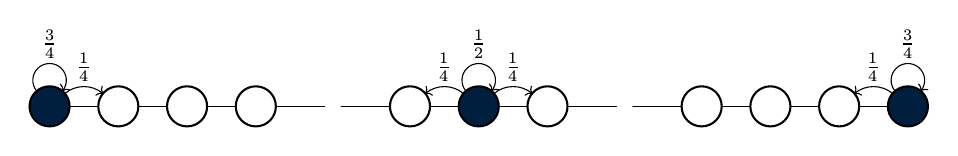}
    \caption{\centering Illustration of lazy symmetry random walk.}
    \label{fig: lazy random walk}
\end{figure}

We introduce the symmetry lazy random walk (SLRW) in polynomial space, or the one-dimensional lattice, to describe the evolution behavior of $L$. A SLRW is a type of Markov process shown in Fig.~\ref{fig: lazy random walk}. In this process, consider a point located at a site $y_i$. In the next time interval, this point has a probability of $0.25$  moving to one of its neighboring sites $y_{i-1}$ or $y_{i+1}$, and it has a probability of $0.5$ staying in place. If the origin site is on the ends of the lattice, it has a probability of $0.75$ staying in place and has a probability of $0.25$ moving around. The probability transition relation can be expressed by Eq.~\eqref{eq: lazy random walk L}. 
We can see that the separate evolution in Eq.~\eqref{eq: separate evolution} fits the form of SLRW. The transition step $\mathcal{L}_{\Gamma_2}$ in the representation space $\Gamma_2'$ can be given by 
\begin{equation}
    \mathcal{L}_{\Gamma_2} \supket{\gamma_S, \gamma_S} = (L\times L)\circ \varphi' \circ \phi (\supket{\gamma_S, \gamma_S}),
    \label{eq: appendix L gamma def}
\end{equation}
where $(L\times L) (y_iy_j) = L(y_i)L(y_j)$. 

In Eq.~\eqref{eq: zz anonymous 5}, we showed the result of how $y_iy_j$ transfers in one specific situation. Now, we will show all possible results in any situation in Table \ref{table: transition table of teto}. The table lists all the possible transition results no matter what inputs it receives. 
It was created in a similar way to the previous example in Eq.~\eqref{eq: zz anonymous 4}.

\begin{table}
    \centering
    \renewcommand\arraystretch{2}
    \begin{tabular}{|c|l|}
        \hline$y_i y_j$ & $ \teto (y_i y_j) $ \\
        \hline$i=j=1$ & $L(y_i)L(y_j) -\frac{5}{144} y_1 y_1-\frac{5}{144} y_2 y_2+\frac{5}{72} y_1 y_2$ \\
        \hline $1<i<N, j=i$ & $L(y_i)L(y_j)-\frac{5}{144} y_{i-1} y_{i-1}+\frac{1}{36} y_{i-1} y_i+\frac{1}                           {24} y_{i-1} y_{i+1}-\frac{1}{36} y_i y_i+\frac{1}{36} y_i y_{i+1}-\frac{5}{144} y_{i+1} y_{i+1}$ \\
        \hline $1 \leq i<N, j=i+1$ & $L(y_i)L(y_j)-\frac{1}{48} y_i y_i-\frac{1}{48} y_{i+1} y_{i+1}+\frac{1}{24} y_i y_{i+1}$ \\
        \hline$i=j=N$ & $L(y_i)L(y_j)-\frac{5}{144} y_N y_N-\frac{5}{144} y_{N-1} y_{N-1}+\frac{5}{72} y_{N-1} y_N$ \\
        \hline other case & $L(y_i)L(y_j)$ \\
        \hline
    \end{tabular}
    \caption{The transition result of $\teto$ with input $y_iy_j$ in different condition. Notice that $y_i y_j = y_j y_i$, the indices of the two factors $i$ and $j$ in term $y_iy_j$ can always be arranged in ascending order $i\leq j$.}
    \label{table: transition table of teto}
\end{table}

In Table \ref{table: transition table of teto}, we can see that in most cases 
\begin{equation}
    \teto(y_i y_j) = L(y_i) L(y_j)
\end{equation}
except for the cases when $|i-j|\leq 1$. 
When $|i-j|\leq 1$, the coefficients of remainder terms $\teto(y_i y_j) - L(y_i)L(y_j)$ are small. 
Denote $L^t(y_i)$ as the outcome of SLRW $t$-steps from $y_i$ according to the propagation rule in Eq.~\eqref{eq: lazy random walk L}, and the $\mathcal{L}_{i}(\mu, t)$ as the probability of stopping at $y_\mu$ after $t$-steps walking,
\begin{equation}
    \begin{aligned}
    L^t(y_i) =& \sum_\mu \mathcal{L}_{i}(\mu, t) y_\mu \\
        \mathcal{L}_{i}(\mu, t) =& \frac{1}{N} + \frac{2}{N} \sum_{k=1}^{N-1} \cos \left(\left(\mu-\frac{1}{2}\right) \frac{\pi k}{N}\right) \cos \left(\left(i-\frac{1}{2}\right) \frac{\pi k}{N}\right) \cos^{2 t} \left(\frac{\pi k}{2 N}\right).
    \end{aligned}
\end{equation}
Refs.~\cite{giuggioli2020exact} gives the analytical solution of $\mathcal{L}_{i}$, which is 
\begin{equation}
    \mathcal{L}_{i}(\mu, t) = \frac{1}{N} + \frac{2}{N} \sum_{k=1}^{N-1} \cos \left(\left(\mu-\frac{1}{2}\right) \frac{\pi k}{N}\right) \cos \left(\left(i-\frac{1}{2}\right) \frac{\pi k}{N}\right) \cos^{2 t} \left(\frac{\pi k}{2 N}\right).
\end{equation}
Thus, for the evolution that can be separated by $\teto(y_iy_j) = L(y_i) L(y_j)$, we can get the analytical solution results
\begin{align}
    \mathcal{C}_{\mathcal{P}_N} (y_i y_j) =& (\teto)^t (y_i y_j) \\
    =& L^t(y_i)L^t(y_j) + R(y_iy_j)\\
    \label{eq: zz anonymous 7}
    =&\sum_{\mu,\nu} \left(\mathscr{L}_{ij} (\mu, \nu, t) y_\mu y_\nu + \mathscr{R}_{ij} (\mu, \nu, t) y_\mu y_\nu\right),
    % \label{eq: zz anonymous 20}
\end{align}
where
\begin{equation}
    \mathscr{L}_{ij} (\mu, \nu, t) := \mathcal{L}_{i}(\mu, t)\mathcal{L}_{j}(\nu, t),
    \label{eq:L_ij_deviation}
\end{equation}
and $R$ stands for the reminder terms caused by the near-diagonal terms $y_i y_j$, $|i-j| \leq 1$ in Table \ref{table: transition table of teto}. 

The computation of \(\alpha_{S,d}\) can be divided into two primary components: one originating from the SLRW and the other arising from the remainder terms. Specifically, \(\alpha_{S,d}\) is determined through the tensor contraction of the operator \(\mathcal{C}\), the state \(\supket{P_S}\), and the state \(\supket{\bm{0}, \bm{0}}\). To facilitate this calculation, both \(\mathcal{C}\) and \(\supket{P_S}\) have been mapped into polynomial space. Similarly, the state \(\supket{\bm{0}, \bm{0}}\) is also represented within this polynomial framework.

Eq.~\eqref{eq: zz anonymous 6} plays a crucial role in this process by transforming \(\supket{\bm{0}, \bm{0}}\) into the Pauli basis. In this context, \(\supbra{Z_i}\) acts as a projection operator that converts a supervector into a scalar. Specifically, \(\supbra{Z_i}\) maps the supervector \(\supket{Z_i}\) to 1 while mapping all other Pauli basis vectors to 0. As previously established, the supervector \(\supket{Z_i}\) is associated with the monomial \(x_i^2\) in the polynomial space \(\mathcal{P}_n\).

Recognizing that the space of derivative operators constitutes the dual space to \(\mathcal{P}_n\), we employ these operators to represent the dual space of \(\mathcal{P}_n\). In this representation, the dual vectors are mapped as 
\begin{equation}
    \supbra{\psi_{ij}} \to \frac{\partial^2}{\partial x_i \partial x_j},~~\supbra{\psi_{ii}}\to \frac{1}{2}\frac{\partial^2}{\partial x_i^2}.
\end{equation}
This mapping ensures that the action of the dual vectors on the polynomial space corresponds to differentiation operations, which are fundamental in analyzing the system's behavior.

Using analogous methods and performing the necessary algebraic manipulations, we extend this mapping to transform \(\supbra{\bm{0}, \bm{0}}\) into the polynomial space \(\mathcal{P}_N\) for the case where \(|S| = 2\). \ gives the resulting expression \begin{equation}
    \supbra{\bm 0,\bm 0} \to \frac{1}{6} \sum_{i} \frac{\partial^2}{\partial y_i^2 }.
\end{equation}
This transformation incorporates factor \(\frac{1}{6}\), which arises from the definitions of the Pauli and polynomial bases. Subsequently, the \(\alpha_{S,2t+1}\) for odd depth \(d = 2t + 1\) can be expressed in terms of these derivative operators and the action of \(\mathcal{C}\) within the polynomial space \(\mathcal{P}_N\)
\begin{equation}
    \alpha_{\{\gamma_i\gamma_j\},2t+1} = \frac{1}{6} \sum_{\mu} \frac{\partial^2}{\partial y_\mu^2 } \mathcal{C}_{\mathcal{P}_N} (y_i y_j).
    \label{eq: zz anonymous 275}
\end{equation}
By combining Equation~\eqref{eq: zz anonymous 7} with Equation~\eqref{eq: zz anonymous 275}, we obtain:
\begin{align}
    \alpha_{\{\gamma_i\gamma_j\},2t+1} =& \frac{1}{3} \sum_\mu \mathscr{L}_{ij} (\mu, \mu, t) + \frac{1}{3} \sum_\mu \mathscr{R}_{ij} (\mu, \mu, t) \\
    =:&  \alpha_{\{\gamma_i\gamma_j\}, 2t+1}^{L} + \alpha_{\{\gamma_i\gamma_j\}, 2t+1}^{R},
    \label{eq: zz anonymous 24}
\end{align}
where \(\mathscr{L}_{ij} (\mu, \mu, t)\) represents a transition probability associated with the SLRW, and \(\mathscr{R}_{ij} (\mu, \mu, t)\) corresponds to the reminders. 

\section{Estimate the order of $\alpha_{\{\gamma_i\gamma_j\}, 2t+1}^{L}$}
\label{sec: estimate the order of alpl}
We have separated the calculation of $\alpha_{\{\gamma_i\gamma_j\},2t+1}$ into two parts in Sec.~\ref{sec: mapping the action of tensors to random walk}. In this section, we aim to estimate the order of the first part, $\alpl$.
\begin{lemma}
\label{theorem: order of alpha l}
    The following formula can estimate the $\alpl$
    \begin{equation}
        \alpl  = \frac{1}{3\sqrt{2\pi t}}  \sum_{k=-\infty}^{\infty} \left(e^{-\frac{(2Nk+a)^2}{2t}} + e^{-\frac{(2Nk+b)^2}{2t}} \right)+\mathcal{O}\left(e^{-\frac{\pi^2}{2}t}\right)
    \end{equation}
    where $a$ is defined as $|i-j|$ and $b$ is defined as $i+j-1$.
\end{lemma}

\begin{proof}
Recall that $N = n/2$. By Lemma~\ref{lemma: simplification of alpha L}, the expression for $\alpha_{S,2t+1}$ can be simplified to the form given in Eq.~\eqref{eq: alpha L}. Moreover, adding the term corresponding to $k=0$ does not affect the summation series,
\begin{equation}
\label{eq: zz anonymous 1387}
\begin{aligned}
3\alpl 
=-\frac{1}{N}+\frac{1}{N} \sum_{k=0}^{N-1}\left[\cos \left((i-j) \frac{k \pi}{N}\right)+\cos \left((i+j-1) \frac{k \pi}{N}\right)\right] \cos ^{4 t}\left(\frac{\pi k}{2 N}\right).
\end{aligned}
\end{equation}
Hence, the summation over $k$ spans a complete cycle, enabling us to leverage certain useful properties of trigonometric summations.

Our goal is to use the Poisson summation formula to estimate the order of $\alpha_{S,d}$. To facilitate this, it is more practical to express the formula using exponentials rather than trigonometric functions. Note that $e^{-2tx^2}$ serves as a good approximation for $\cos^{4t}(x)$,
\begin{equation}
\begin{aligned}
e^{-2 t x^2}-\cos ^{4 t}(x) & =  e^{-2 t x^2}-e^{-2 t x^2+O\left(t x^4\right)} \\
&=  e^{-2 t x^2}\left(1-e^{O\left(t x^4\right)}\right) \\
&\sim \mathcal{O}\left(t x^4 e^{-2 t x^2}\right).
\end{aligned}
\end{equation}

By substituting $\cos ^{4 t}\left(\frac{\pi k}{2 N}\right)$ with $e^{-\frac{k^2 \pi^2 t}{2N^2}}$ in Eq.~\eqref{eq: zz anonymous 1387}, we have 
\begin{align}
3\alpl& =-\frac{1}{N}+\frac{1}{N} \sum_{k=0}^{N-1} e^{-\frac{k^2 \pi^2 t}{2 N^2}}\left[\cos \left((i-j) \frac{k \pi}{N}\right)+\cos \left((i+j-1) \frac{k \pi}{N}\right)\right]+\mathcal{O}\left(e^{-\frac{\pi^2}{2}t}\right)\\
& =-\frac{1}{N}+\frac{1}{N} \sum_{k=0}^{\infty}e^{-\frac{k^2 \pi^2 t}{2 N^2}}\left[\cos \left((i-j) \frac{k \pi}{N}\right)+\cos \left((i+j-1) \frac{k \pi}{N}\right)\right]+\mathcal{O}\left(e^{-\frac{\pi^2}{2}t}\right).
\label{eq: zz anonymous 19}
\end{align}
In Eq. \eqref{eq: zz anonymous 19}, we expand the summation to infinity, and it will not introduce much of errors because 
\begin{align*}
    \sum_{k=N}^{\infty} e^{-\frac{k^2 \pi^2 t}{2 N^2}} =& e^{-\frac{\pi^2}{2}t}\sum_{k=0}^{\infty} e^{-\frac{k^2 \pi^2 t}{2 N^2}} \\
    \leq& e^{-\frac{\pi^2}{2}t}\sum_{k=0}^{\infty} e^{-\frac{k \pi^2 t}{2 N^2}} \\
    =& e^{-\frac{\pi^2}{2}t} \frac{e^{\frac{\pi^2 t}{2 N^2}}}{e^{\frac{\pi^2 t}{2 N^2}} -1 } \\
    =& \mathcal{O}\left(e^{-\frac{\pi^2}{2}t}\right)
\end{align*}

According to the Poison summation formula, 
\begin{equation}
     \sum_{k=-\infty}^{\infty} e^{-\frac{k^2 \pi^2 t}{2 N^2}}\cos \left((i-j) \frac{k \pi}{N}\right) 
    = N\sqrt{\frac{2}{\pi t}}\sum_{k=-\infty}^{\infty} e^{-\frac{(2Nk+i-j)^2}{2t}}
    \label{eq: zz 589}
\end{equation}
By substituting Eq.~\eqref{eq: zz 589} into Eq.~\eqref{eq: zz anonymous 19}, we have
\begin{align}
\alpl& = \frac{1}{3\sqrt{2\pi t}}  \sum_{k=-\infty}^{\infty} \left(e^{-\frac{(2Nk+i-j)^2}{2t}} + e^{-\frac{(2Nk+i+j-1)^2}{2t}} \right)+\mathcal{O}\left(e^{-\frac{\pi^2}{2}t}\right).
\end{align}

\end{proof}

In the following, we give the relationship between $\alpha_{S,d}^L$ with the depth of random matchgate circuits.
\begin{lemma}
    \label{corollary: 1}
    The $\alpha_{S,d}^L$ scales $\mathcal{O}\left( \frac{1}{\mathrm{poly}(n)}\right)$ if $d = \Theta\left(\max\left\{ \frac{d_{\text{int}}(S)^2}{\log(n)} , d_{\text{int}}(S)\right\}\right)$.
\end{lemma}
\begin{proof}
    We start by proving that
    \begin{equation}
    \label{eq: appendix corollary 1 1}
        \sum_{k=q+1}^{\infty} e^{-\frac{(2Nk+i-j)^2}{2t}}  = \mathcal{O}\left(e^{-\frac{(2Nq+i-j)^2}{2t}}\right)
    \end{equation}
    for positive integer $q$:
    \begin{align}
        \frac{\sum_{k=q+1}^{\infty} e^{-\frac{(2Nk+i-j)^2}{2t}}}{e^{-\frac{(2Nq+i-j)^2}{2t}}} =& \sum_{k=q+1}^{\infty} \exp\left[\frac{(2Nq+i-j)^2}{2t}-\frac{(2Nk+i-j)^2}{2t} \right]\\
        =& \sum_{k=q+1}^{\infty} \exp\left[-\frac{2N(k-q)(N(q+k) -(i-j))}{t}\right] \\
        \leq& \sum_{k=q+1}^{\infty} \exp\left[-\frac{2N(k-q)(2Nq -(i-j))}{t}\right],
    \end{align}
    where $\sum_{k=q+1}^{\infty} \exp\left[-\frac{2N(k-q)(2Nq -(i-j))}{t}\right]$ converges to certain constant. Thus, we prove Eq.~\eqref{eq: appendix corollary 1 1}.

    Eq.~\eqref{eq: appendix corollary 1 1} indicates that the order of $\alpha_{S,d}^L$ is determined by $\frac{1}{\sqrt{2\pi t}}  e^{-\frac{(i-j)^2}{2t}}+\mathcal{O}\left(e^{-\frac{\pi^2}{2}t}\right).$
    The condition $d = 2t+1 = \Omega(d_{\text{int}})$ ensures that $\mathcal{O}\left(e^{-\frac{\pi^2}{2}t}\right) =\mathcal{O}\left(\frac{1}{\mathrm{poly}(n)}\right) $, while the condition $d = \Omega(\frac{d_{\text{int}}(S)^2}{\log(n)} )$ ensures that $e^{-\frac{(i-j)^2}{2t}} =\Omega\left(\frac{1}{\mathrm{poly}(n)}\right) $ (notice that $d_\text{int}(S) = |i-j|$ for 2-local Majorana $\gamma_S$). Thus, we conclude that $\alpha_{S,d}^L = \Omega\left( \frac{1}{\mathrm{poly}(n)}\right)$ when $d = \Theta\left(\max\left\{ \frac{d_{\text{int}}(S)^2}{\log(n)} , d_{\text{int}}(S)\right\}\right)$. 
    
\end{proof}

\begin{lemma}
\label{lemma: simplification of alpha L}
The expression of $\alpl$ can be simplified in the following form
    \begin{equation}
    \label{eq: alpha L}
        3\alpl=\frac{1}{N}+\frac{1}{N} \sum_k\left[\cos \left((i-j) \frac{k \pi}{N}\right)+\cos \left((i+j-1) \frac{k \pi}{N}\right)\right] \cos ^{4 t}\left(\frac{\pi k}{2 N}\right).
    \end{equation}
\end{lemma}
\begin{proof}
By the definition of $\alpl$ and Eq. \ref{eq:L_ij_deviation}, we have
\begin{align}
    3\alpl =& \sum_\mu \left[ \frac{1}{N} + \frac{2}{N} \sum_{k=1}^{N-1} \cos \left(\left(i-\frac{1}{2}\right) \frac{\pi k}{N}\right) \cos \left(\left(\mu-\frac{1}{2}\right) \frac{\pi k}{N}\right) \cos ^{2 t} \left( \frac{\pi k}{2 N} \right)\right] \nonumber\\
     &\times\left[ \frac{1}{N} + \frac{2}{N} \sum_{l=1}^{N-1} \cos \left(\left(j-\frac{1}{2}\right) \frac{\pi l}{N}\right) \cos \left(\left(\mu-\frac{1}{2}\right) \frac{\pi l}{N}\right) \cos ^{2 t}\left( \frac{\pi l}{2 N} \right)  \right] \nonumber\\
     =& 1/N + \frac{2}{N^2} \sum_k \left[\sum_\mu \cos\left(\left(\mu-\frac{1}{2}\right) \frac{\pi k}{N}\right)  \right] 
     \cos \left(\left(i-\frac{1}{2}\right) \frac{\pi k}{N}\right) \cos ^{2 t} \left( \frac{\pi k}{2 N} \right) 
     \label{eq: zz anonymous 9}
     \\
     &+ \frac{2}{N^2} \sum_l \left[\sum_\mu \cos\left(\left(\mu-\frac{1}{2}\right) \frac{\pi l}{N}\right)  \right] 
     \cos \left(\left(i-\frac{1}{2}\right) \frac{\pi l}{N}\right) \cos ^{2 t} \left( \frac{\pi l}{2 N} \right) 
     \label{eq: zz anonymous 10}
     \\
     &+ \frac{4}{N^2} \sum_{k,l =1}^{N-1} \sum_{\mu = 1}^{N} \cos \left(\left(i-\frac{1}{2}\right) \frac{\pi k}{N}\right) \cos \left(\left(\mu-\frac{1}{2}\right) \frac{\pi k}{N}\right) \cos ^{2 t} \left( \frac{\pi k}{2 N} \right)  \nonumber\\
     &\times \cos \left(\left(j-\frac{1}{2}\right) \frac{\pi l}{N}\right) \cos \left(\left(\mu-\frac{1}{2}\right) \frac{\pi l}{N}\right) \cos ^{2 t}\left( \frac{\pi l}{2 N} \right)  
     \label{eq: calculation of alpha L}
\end{align}

Notice that the summation of cosine function is zero 
\begin{equation}
\begin{aligned}
\sum_{\mu=1}^{N} \cos\left(\left(\mu-\frac{1}{2}\right) \frac{\pi k}{N}\right)& =-\frac{1}{2} \cos \left(\frac{1}{2} \pi(2 k+1)\right) \csc \left(\frac{\pi k}{2 N}\right) \\
& =\sin (k \pi) \csc \left(\frac{\pi k}{2 N}\right) \\
& =0.
\end{aligned}
\end{equation}
Substitute this identity into Eq.~\eqref{eq: calculation of alpha L}, we can eliminate the terms in Eqs. \eqref{eq: zz anonymous 9} and \eqref{eq: zz anonymous 10}. 

Also, notice that
\begin{equation}
\label{eq: zz anonymous 12}
\begin{aligned}
& \sum_{\mu=1}^{N} \cos \left(\frac{\pi\left(\mu-\frac{1}{2}\right) i}{N}\right) \cos \left(\frac{\pi\left(\mu-\frac{1}{2}\right) j}{N}\right) \\
= & \frac{1}{2} \sum_{\mu=1}^{N} \cos \left(\frac{\pi\left(\mu-\frac{1}{2}\right) i}{N}-\frac{\pi\left(\mu-\frac{1}{2}\right) j}{N}\right)+\cos \left(\frac{\pi\left(\mu-\frac{1}{2}\right) i}{N}+\frac{\pi\left(\mu-\frac{1}{2}\right) j}{N}\right) \\
= & \frac{1}{2} \sum_{\mu=1}^{N} \cos \left(\frac{\pi(2 \mu-1)(i-j)}{2 N}\right)+\cos \left(\frac{\pi(2 \mu-1)(i+j)}{2 N}\right) \\
= & \frac{1}{4}\left(\sin (\pi(i+j)) \csc \left(\frac{\pi(i+j)}{2 N}\right)-\sin (\pi(j-i)) \csc \left(\frac{\pi(i-j)}{2 N}\right)\right).
\end{aligned}
\end{equation}
This result gets value $0$ when $\frac{\pi(i+j)}{2 N}\neq a \pi$ or $\frac{\pi(i-j)}{2 N}\neq b \pi$ for some integer $a$ and $b$, because $\sin(\pi m) = 0$. The term $ \sin(\pi m) \csc(\frac{\pi m}{2N})$ gets non-zero only when $\csc(\frac{\pi m}{2N})$ gets infinity. Then, we can write down the conditions that $i$ and $j$ satisfy
\begin{equation}
\left\{\begin{array}{l}
i+j=2 a N \text { or } |i-j|=2 b N \\
a, b \in \mathbb{Z} \\
1<i,j<N-1.
\end{array} \right.
\end{equation}
The equation shows that the result is $j=k$. We can use  L'Hôpital's rule to calculate the term
\begin{equation}
    \label{eq: zz anonymous 11}
    \lim_{x\to 0} \sin(\pi x) \csc(\frac{\pi x}{2N}) = 2N
\end{equation}
when $i$ and $j$ satisfy the condition $i=j$. Plugin Eq.~\eqref{eq: zz anonymous 11} and Eq.~\eqref{eq: zz anonymous 12} into Eq.~\eqref{eq: calculation of alpha L}, we have
\begin{equation}
    \label{eq: zz anonymous 13}
    3\alpl = \frac{1}{N} + \frac{2}{N} \sum_k \cos \left(\left(i-\frac{1}{2}\right) \frac{\pi k}{N}\right)\cos \left(\left(j-\frac{1}{2}\right) \frac{\pi k}{N}\right)\cos ^{4 t}\left( \frac{\pi k}{2 N} \right).
\end{equation}
Finally, we use trigonometric identities to expand this equation, thereby completing this proof
\begin{equation}
    3\alpl=\frac{1}{N}+\frac{1}{N} \sum_k\left[\cos \left((i-j) \frac{k \pi}{N}\right)+\cos \left((i+j-1) \frac{k \pi}{N}\right)\right] \cos ^{4 t}\left(\frac{\pi k}{2 N}\right).
\end{equation}
\end{proof}

\section{The relation between $\alpl$ and $\alpr$}
\label{appendix relation between alphaL and alpha}
Recall that we have divided the calculation of $\alpha_{\{\gamma_i\gamma_j\}, 2t+1}$ into two parts. One is the $\alpl$ and the other is $\alpr$. Lemma \ref{theorem: order of alpha l} gives the order of $\alpl$. In this section, we aim to bound the $\alpr$ by $\alpl$, so that the order of $\alpha_{\{\gamma_i\gamma_j\}, 2t+1}$ can be given by the $\alpl$.

We begin with the polynomial in Eq.~\eqref{eq: zz anonymous 7}
\begin{equation}
\label{eq: zz anonymous 21}
    \mathcal{C}_{\mathcal{P}_N} (t) (y_i y_j) 
    =\sum_{\mu,\nu} \left(\mathscr{L}_{ij} (\mu, \nu, t) y_\mu y_\nu + \mathscr{R}_{ij} (\mu, \nu, t) y_\mu y_\nu\right).
\end{equation}
Notice that the $\mathscr{R}_{ij}$ is the corresponding term of $\Delta$ (see the main text, Eq.~\eqref{eq: def del}) in the polynomial space $\mathcal{P}_N$
Then, we let the polynomial transform one time-interval step, and we get
\begin{align}
    &\mathcal{C}_{\mathcal{P}_N} (t+1) (y_i y_j) \\
    =& \teto\left(\mathcal{C}_{\mathcal{P}_N} (t) (y_i y_j)\right) \\
    =& \sum_{\mu,\nu} \left(\mathscr{L}_{ij} (\mu, \nu, t) (L(y_\mu) L(y_\nu) + R(y_\mu, y_\nu)) + \mathscr{R}_{ij} (\mu, \nu, t) (L(y_\mu) L(y_\nu) + R(y_\mu, y_\nu))\right)\\
    =& \sum_{\mu,\nu} \left(\mathscr{L}_{ij} (\mu, \nu, t) (L(y_\mu) L(y_\nu) + R(y_\mu, y_\nu)) + \mathscr{R}_{ij} (\mu, \nu, t) (L(y_\mu) L(y_\nu) + R(y_\mu, y_\nu))\right)\\
    =& \sum_{\mu,\nu} \mathscr{L}_{ij} (\mu, \nu, t+1) y_\mu y_\nu +  \sum_{\mu,\nu}\mathscr{L}_{ij} (\mu, \nu, t) R(y_\mu, y_\nu) \nonumber\\
    & + \sum_{\mu,\nu}\mathscr{R}_{ij} (\mu, \nu, t) (L(y_\mu) L(y_\nu) + R(y_\mu, y_\nu)).
    \label{eq: zz anonymous 22}
\end{align}
Deduce from Eq.~\eqref{eq: zz anonymous 21}, we have
\begin{equation}
    \label{eq: zz anonymous 23}
    \mathcal{C}_{\mathcal{P}_N} (t+1) (y_i y_j) 
    =\sum_{\mu,\nu} \left(\mathscr{L}_{ij} (\mu, \nu, t+1) y_\mu y_\nu + \mathscr{R}_{ij} (\mu, \nu, t+1) y_\mu y_\nu\right). 
\end{equation}
Compare Eq.~\eqref{eq: zz anonymous 22} and Eq.~\eqref{eq: zz anonymous 23}, we have 
\begin{align}
    &\sum_{\mu,\nu} \mathscr{R}_{ij} (\mu, \nu, t+1) y_\mu y_\nu = \sum_{\mu,\nu}\mathscr{L}_{ij} (\mu, \nu, t) R(y_\mu, y_\nu) + \sum_{\mu,\nu}\mathscr{R}_{ij} (\mu, \nu, t) (L(y_\mu) L(y_\nu) + R(y_\mu, y_\nu)) \\
    (1+ \delta_{l,k})  &\mathscr{R}_{ij} (l, k, t+1) =  \sum_{\mu,\nu}\mathscr{R}_{ij} (\mu, \nu, t) \frac{\partial^2L(y_\mu) L(y_\nu)}{\partial y_l \partial y_k} + \sum_{\mu,\nu}\left( \mathscr{L}_{ij} (\mu, \nu, t) + \mathscr{R}_{ij} (\mu, \nu, t) \right)\frac{\partial^2 R(y_\mu, y_\nu)}{\partial y_l \partial y_k}.
    \label{eq: strict recursive relation}
\end{align}
Eq.~\eqref{eq: strict recursive relation} describes the strict relationship between $\mathscr{R}_{ij}$ and $\mathscr{L}_{ij}$ in a recursive form, thereby giving the relationship between $\alpl$ and $\alpr$. However, deriving the general term formula from this recursive formula is difficult. Therefore, we hope to use some inequalities to simplify this recursive relationship and thus bound $\alpr$ by $\alpl$.

\begin{figure}[h]
    \centering
    \includegraphics[width=\linewidth]{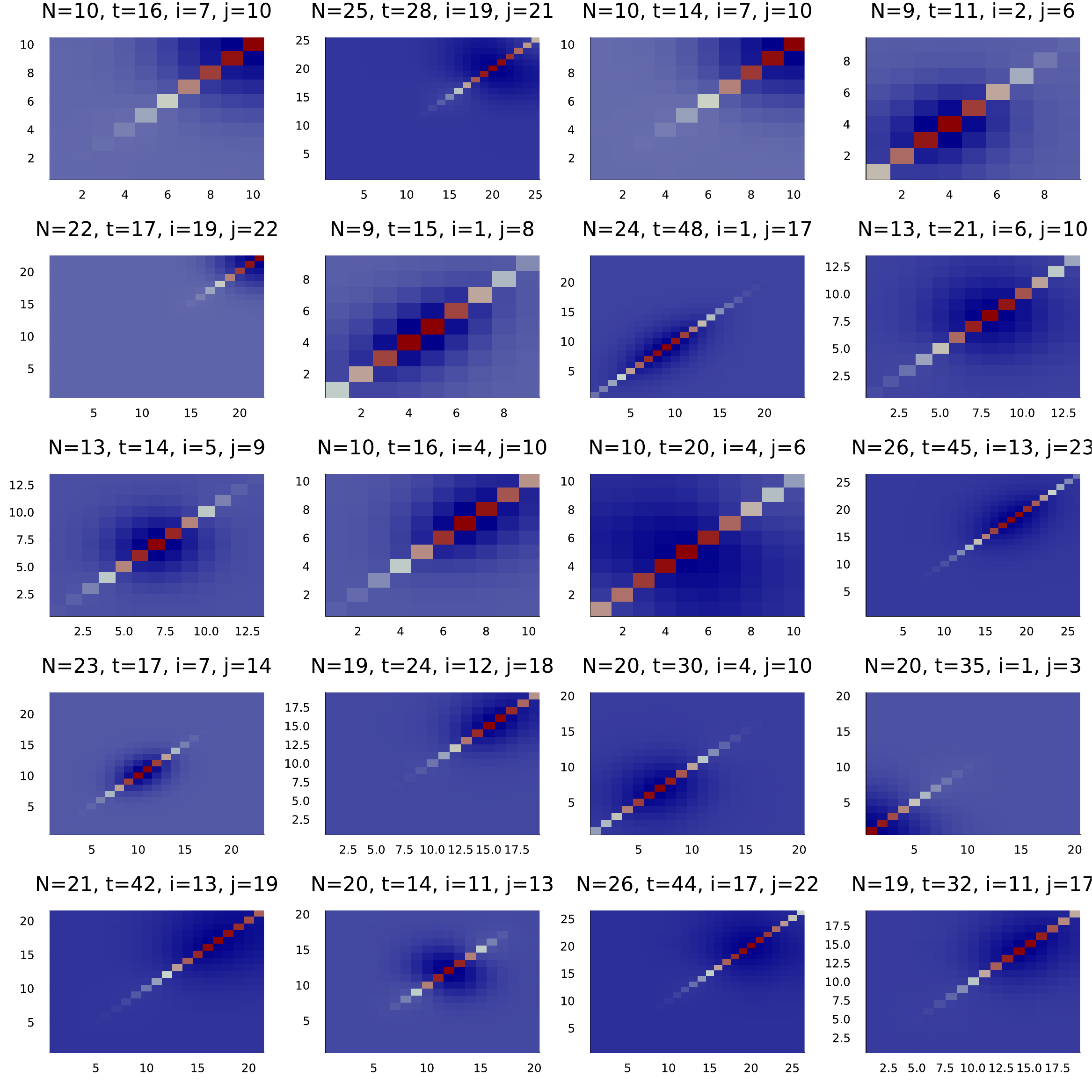}
    \caption{\centering Numerical evaluation of  $\mathscr{R}_{i, j}(\mu, \nu, t)$ by randomly choosing $N, t, i, j$. The pixel is colored red if $\mathscr{R}_{i, j}(\mu, \nu, t)$ is smaller than 0 and blue otherwise. The x-axis represents the value of $\mu$, and the y-axis represents the value of $\nu$. The results indicate that $\mathscr{R}_{i, j}(\mu, \mu, t)<0$ when $\mu\neq\nu$.}
    \label{fig: appendix assumption}
\end{figure}

We will first define some auxiliary variables, 
\begin{align*}
\zeta_k(t) & :=\frac{1}{2} \sum_{\mu=1}^{N-k}\left(\mathcal{L}_i(\mu, t) \mathcal{L}_j(\mu+k, t)+\mathcal{L}_i(\mu+k, t) \mathcal{L}_j(\mu, t)\right) \\
\beta_k(t) & :=\frac{1}{2} \sum_{\mu=1}^{N-k}\left(\mathscr{R}_{i, j}(\mu, \mu+k, t)+\mathscr{R}_{i, j}(\mu+k, \mu, t)\right) \\
a(t) & :=\binom{\zeta_0(t)}{\zeta_1(t)}, \quad b(t):=\binom{\beta_0(t)}{\beta_1(t)}.
\end{align*}
Notice that $3\alpl = \zeta_0(t)$, $3\alpr = \beta_0(t)$. Because $\mathscr{L}_{i, j}(\mu, \nu, t) + \mathscr{R}_{i, j}(\mu, \nu, t)$ represents the probability of being in site $y_iy_j$ during a random walk, it satisfies the property that the summation across all sites is $1$. Meanwhile, the summation of all $\mathscr{L}_{i, j}(\mu, \nu, t)$ is $1$. The two things deduce that
\begin{equation}
    \sum_{\mu, \nu} \mathscr{R}_{i, j}(\mu, \nu, t) =0.
\end{equation}
Especially, in numerical simulation, we observe that all $\mathscr{R}_{i, j}(\mu, \nu, t)$ are greater than $0$ except for $\mu=\nu$. The numerical results are shown in Fig.~\ref{fig: appendix assumption}.
Under this assumption, we have the following lemma:
\begin{lemma}
\label{theorem: relation between alpha L and alpha R 2}
    Assume that $\forall \mu\neq\nu$, $\mathscr{R}_{i, j}(\mu, \nu, t)\geq 0$, and $\forall \mu \neq \nu$, $\mathscr{R}_{i, j}(\mu, \mu, t)\leq 0$. $-\beta_0(t) \leq \frac{25}{72} \max_{k \geq 0}\left\{\zeta_0(t-k)\right\}$.
\end{lemma}

Lemma \ref{theorem: relation between alpha L and alpha R 2} establishes the mathematical relationship between the two components $\alpl$ and $\alpr$. By combining this theoretical relationship with the equality described in Eq.~\eqref{eq: zz anonymous 24}, we can deduce the order of magnitude of the $\alpha_{\{\gamma_i\gamma_j\}, 2t+1}$. Combining these pieces allows us to systematically determine the scale or size of $\alpha_{\{\gamma_i\gamma_j\}, 2t+1}$ based on the other defined quantities.

\begin{proof}[proof of Lemma \ref{theorem: relation between alpha L and alpha R 2}]
   From the recursive relationship in Eq.~\eqref{eq: strict recursive relation} and Table \ref{table: transition table of teto}, we can write down the recursive relationship of $\beta_k$
   \begin{align}
    \beta_0(t+1) & \geq \frac{6}{16} \beta_0(t)+\frac{8}{16} \beta_1(t) + \frac{2}{16}\beta_2(t)-\frac{14}{144}\left(\beta_0(t)+\zeta_0(t)\right)-\frac{1}{24}\left(\beta_1(t)+\zeta_1(t)\right)\nonumber\\
    & + \frac{4}{16}(\mathscr{R}_{i j}(0,0, t-1) + \mathscr{R}_{i j}(N,N, t-1)) -  \frac{1}{16}(\mathscr{R}_{i j}(0,1, t-1) + \mathscr{R}_{i j}(1,0, t-1))\nonumber\\
    &- \frac{1}{16}(\mathscr{R}_{i j}(N-1,N, t-1) + \mathscr{R}_{i j}(N,N-1, t-1))\nonumber\\
    & \geq \frac{5}{18} \beta_0(t)+\frac{11}{24} \beta_1(t)-\frac{7}{72} \zeta_0(t)-\frac{1}{12} \zeta_1(t) \\
    & \geq \frac{5}{18} \beta_0(t)+\frac{5}{24} \beta_1(t)-\frac{7}{72} \zeta_0(t)-\frac{1}{12} \zeta_1(t) 
    \end{align}
    Here, we use the property that $\beta_0(t)+\zeta_0(t)$ is the summation of properties so that it is greater than $0$ to inequality deflate the terms at the edges, like $y_0y_0$ or $ y_0y_1$. Recall that we hold the assumption that $\mathscr{R}_{i, j}(\mu, \mu, t)<0$, and $\forall \mu \neq \nu$, $\mathscr{R}_{i, j}(\mu, \mu, t)>0$, so the edges terms of $\mathscr{R}_{i j}$ can be deflated out as well. Similarly, we have
    \begin{equation}
        \beta_1(t+1)  \geq \frac{5}{9} \beta_0(t)+\frac{5}{12} \beta_1(t)+\frac{1}{18} \zeta_0(t)+\frac{1}{24} \zeta_1(t).
    \end{equation}
    
    Let $\beta_0'$ and $\beta_1'$ obtains the above recursive relation
    \begin{equation}
        \begin{cases}
            \beta_0'(t+1) =& \frac{5}{18} \beta_0(t)+\frac{5}{24} \beta_1(t)-\frac{7}{72} \zeta_0(t)-\frac{1}{12} \zeta_1(t) \\
        \beta_1'(t+1)  =& \frac{5}{9} \beta_0(t)+\frac{5}{12} \beta_1(t)+\frac{1}{18} \zeta_0(t)+\frac{1}{24} \zeta_1(t)
        \end{cases}
        \label{eq: zz anonymous 25}
    \end{equation}
    with the same first term $\beta_0'(0) = \beta_0(0)$ and $\beta_1'(0) = \beta_1(0)$. The $\beta$ and $\beta'$ satisfy the relationship
    \begin{equation}
        \beta_0'(t) \geq \beta_0(t),\quad  \beta_1'(t) \geq \beta_1(t).
    \end{equation}
    Similarly, we denote $b'$ as $b'(t):=\binom{\beta_0'(t)}{\beta_1'(t)}$,

    We rewrite the inequality groups \eqref{eq: zz anonymous 25} to the matrix form
    \begin{align}
        b'(t+1)=C_b b'(t)+C_a a(t), \quad \text{where}\\
        C_b=\left(\begin{array}{c c}
        \frac{5}{18} & \frac{5}{24} \\ 
        \frac{5}{9} & \frac{5}{12}
        \end{array}\right), \quad C_a=\left(\begin{array}{cc}
        -\frac{7}{72} & -\frac{1}{12} \\
        \frac{1}{18} & \frac{1}{24}
        \end{array}\right),
    \end{align}
    and the matrices $C_b$ and $C_a$ govern these recursive dynamics. To solve this recursion explicitly, we diagonalize the $C_b$ matrix. This allows us to express the recursion in closed form,
    \begin{align}
    \end{align}
The term $C_b^t$ can be calculated by eigenvalue decomposition $C_b=Q \Lambda Q^{-1}$, where $\Lambda=\text{diag}(\frac{25}{36},0)$. The eigenvector corresponding to $\frac{25}{36}$ is $(2, \frac{3}{2})^T$, where $T$ denotes the transpose of vector. This allows us to express $C_b^k$ and $C_b^kC_a$ in terms of eigenvalues and eigenvectors
\begin{align}
    C_b^k=&\left(\frac{5}{6}\right)^{2k}\begin{pmatrix}
2 & \frac{3}{2} \\
0 & 0
\end{pmatrix}\\
C_b^kC_a =& \frac{1}{5}\left(\frac{5}{6}\right)^{2k}\begin{pmatrix}
-\frac{1}{9} & -\frac{5}{48} \\
0 & 0
\end{pmatrix}
\end{align}
for any $k>0$. 

The variables we care about are $\zeta_0$ and $\beta_0$ because they are directly related to the $\alpl$ and $\alpr$. Thus, we mainly consider the first item of $C_b^t b'(0)$ and $C_b^kC_aa(t-k-1)$, which can be expressed in the following form
\begin{equation}
    C_b^kC_aa(t-k-1) = \lambda_k \zeta_0(t-k-1)+\eta_k \zeta_1(t-k-1).
\end{equation}
The $\lambda_k$ and the $\eta_k$ are coefficients 
\begin{equation}
    \begin{array}{ll}
\lambda_k=-\frac{1}{45}\left(\frac{5}{6}\right)^{2 k}, & \eta_k=-\frac{1}{48}\left(\frac{5}{6}\right)^{2 k}, \quad \forall k>0,  \\
\lambda_0=-\frac{7}{72}, & \eta_0=-\frac{1}{12}.
\end{array}
\end{equation}

Building on the previous relationships, we can now derive an explicit formula for $\beta'_0(t)$. From the expression for the first element of $C_b^kC_aa(t-k-1)$, we obtain
\begin{equation}
    \beta_0^{\prime}(t)=\sum_{k=0}^t\left(\lambda_k \zeta_0(t-k-1)+\eta_k \zeta_1(t-k-1)\right)
    \label{eq: zz anonymous 26}
\end{equation}
We also know from the recursive relation of $\mathcal{L}_i$ that the state variables satisfy
\begin{equation}
\zeta_0(t+1) \geq \frac{3}{8} \zeta_0(t)+\frac{1}{2} \zeta_1(t).
\end{equation}
Leveraging this inequality into Eq.~\ref{eq: zz anonymous 26} allows us to place an upper bound on $\beta'_0(t)$
\begin{align}
-\beta_0^{\prime}(t) \leq & -\sum_{k=1}^t\left(\lambda_k-\frac{3}{4} \eta_k+2 \eta_{k+1}\right) \zeta_0(t-k-1)+ \frac{19}{72} \zeta_0(t) - \frac{1}{16}\zeta_0(t-1) \\
\leq & \sum_{k=1}^t\left(\frac{5}{6}\right)^{2 k} \frac{307}{8640} \zeta_0(t-k-1)+ \frac{19}{72} \zeta_0(t) - \frac{1}{16}\zeta_0(t-1) \\
\leq & \frac{25}{72} \max_{k \geq 0}\left\{\zeta_0(t-k)\right\}
\end{align}

\end{proof}

Combine Lemma \ref{theorem: order of alpha l} and Lemma \ref{theorem: relation between alpha L and alpha R 2}, we conclude the following theorem:
\begin{theorem}
    \label{theorem 1}
    The sample complexity scales $\mathcal{O}(\frac{n}{\epsilon^2 })$  up to a log factor for 2-local Majorana strings in the average of $\rho$ when the depth satisfies 
    \begin{equation}
        d^* = \Theta\left(\max\left\{ \frac{d_{\text{int}}(S)^2}{\log(n)} , d_{\text{int}}(S)\right\}\right)
    \end{equation}
    under the following assumption: Compared to the SLRW, the true random walk in polynomial space $\mathcal{P}_N$ (The random walk of $\varphi'\circ\phi (\gamma_S)$ under transition $\mathcal{B}$) has less probability to walks on the diagonal sites $y_iy_j$ but has larger probability to walks on the off-diagonal sites at step $t$. 
\end{theorem}
\begin{proof}
    As we have discussed, the sample complexity is decided by $\alpha_{S,d}$, which can be expanded by 
    \begin{equation}
        \alpha_{S,d} = \alpha_{S,d}^L + \alpha_{S,d}^R.
    \end{equation}
    The assumption about the random walk means 
    \begin{equation}
        \begin{cases}
            \mathscr{R}_{i, j}(\mu, \mu, t) \leq 0, \\
            \mathscr{R}_{i, j}(\mu, \nu, t) \geq 0, \forall \mu\neq\nu.
        \end{cases}
    \end{equation}
    Following Lemma \ref{theorem: relation between alpha L and alpha R 2}, we have 
    \begin{align}
        \alpha_{S,d} =& \alpha_{S,d}^L + \alpha_{S,d}^R\\
            \geq & \alpha_{S,d}^L - \frac{25}{72}\max_{d'<d}\{\alpha^L_{S,d'}\} \\
            \geq & \frac{47}{72}\max_{d'\leq d}\{\alpha^L_{S,d'}\}.
    \end{align}
    If $a = \mathcal{O}(\log n)$, the proof is complete by Lemma \ref{lemma: 7} with depth $d^* = \Theta(|i-j|)$. For $a = \omega(\log n)$, we have 
  $\mathcal{O}(e^{-\pi^2 d}) = \mathcal{O}(n^{-\pi^2 }) $, and $\frac{1}{\sqrt{2\pi t}}e^{-\frac{(i-j)^2}{2t}} = \Omega(n) $ up to a log factor. In this case, we show that $\alpha_{S,d}^L = \Omega(n)$ with a similar method used in Lemma \ref{corollary: 1}. 
  Combine Chebyshev inequality and inequality \eqref{eq: variance and alpha}, and we conclude that the sample complexity is $ \mathcal{O}(\frac{n}{\epsilon^2 }) $. 
\end{proof}

\section{Efficiency when the distance of set is short}
\label{appendix efficiency}
\begin{lemma}
\label{lemma: 7}
The expectation value of $\tr(\rho \gamma_S)$ can be obtained by using $\mathcal{O}(\log(n))$-depth matchgate circuit within the Fermionic classical shadows protocol when the distance of $S$ is $\mathcal{O}(\log n)$ and the cardinal number $|S| = 2k$ is a constant. 
\end{lemma}

\begin{proof}
    Let the initial tensor be $P_S$, and apply $\mathcal{C}$ to the tensor $P_S$. Each $\mathcal{B}$ in $\mathcal{C}$ will transform the $P_S$ to the superposition of a series of Pauli tensors
    \begin{equation}
        \mathcal{B} \supket{P_{S}, P_{S}} = \sum_{|S'| = |S|} \xi_{S'} \supket{P_{S'}, P_{S'}},
    \end{equation}
    where $\xi_{S'}$ are real coefficients that satisfy $\sum \xi_{S'} = 1$, and the concrete number of $\xi_{S'}$ are shown in Table \ref{table: the 16x16 matrix of T}.
    The tensor network contraction of $\alpha_{S,d}$ contracts out the 
    outputs supervectors $ \supket{\mathcal{Z}}$. The supervectors $\supket{P_{S'}, P_{S'}}$ in the output state $\mathcal{C} \supket{P_{S}, P_{S}} = \sum_{|S'| = |S|} \xi_{S'}^{C} \supket{P_{S'}, P_{S'}}$ make non-zero contribution to the $\alpha_{S,d}$ only if it just contains Pauli $\mathbbm{1}$ or $I$ operators. Thus, at each step, we only preserve the branches $S'$ which has the trend to arrive $\supket{P_{S'}, P_{S'}}$ with $\mathbbm{1}$ or $I$. Formally, for each transition $\mathcal{B}$, we follows:
    \begin{enumerate}
        \item preserves $\supket{P_{S'}, P_{S'}}$ if $d_\text{near}(S')\leq d_\text{near}(S)$
        \item discards $\supket{P_{S'}, P_{S'}}$ (let $\xi_{S'} = 0$) if $d_\text{near}(S')> d_\text{near}(S)$
     \end{enumerate}
    where $d_\text{near}(S):= \max\{ i_{2j} - i_{2j-1} \mid j \in [n] \}$.

    Table ~\ref{table: the 16x16 matrix of T} tells us two facts:
    \begin{enumerate}
        \item There exist a $S'$ such that $d_\text{near}(S')= d_\text{near}(S)-m$ for $S$ with $d_\text{near}(S) \geq m$, $m\leq 2$.
        \item the summation of $\xi_{S'}$ for remaining branches is greater than $\frac{1}{36^{k}}$. 
    \end{enumerate}
    Thus, with $d_\text{near}(S)/2$ steps, there exists a $S'$ such that the $\xi_{S'}$ is non-zero. Since $d_\text{int}(S)\geq d_\text{near}(S)$, the summation of coefficients of the remaining branches is greater than $\frac{1}{36^{k d_\text{int}(S)/2}}$ after $d_\text{int}(S)/2$ steps transition. Notice that $|S| = k$ is a constant number and $d_\text{int}(S) = \mathcal{O}(\log(n))$, the summation number is 
    \begin{equation}
        \sum_{S''} \xi_{S''} \geq  \frac{1}{36^{k d_\text{int}(S)/2}} = \frac{1}{\mathcal{O}(n^{k/2} \polylog(n))}
    \end{equation}
    where $S''$ are the remaining branches.
\end{proof}

\section{More data for numerical simulation}

Fig.~\ref{fig: fix dint} shows the values of $\alpha_{S,d}$ and $\alpha_{S,d}'$ with the same $d_\text{int}$ and different $S$. 
The optimal depth $d^\ast$ depends on the $d_\text{int}$, which means the four different curves share the same optimal depth order. 
Although the values $\alpha_{S,d}'$ are different for different $S$, depth $d^\ast$ promises that for all sets $S$ with the same $d_\text{int}$ scales polynomially small.

\label{appendix: numerical simulation}

\begin{figure}
    \centering
    \includegraphics[width=0.8\linewidth]{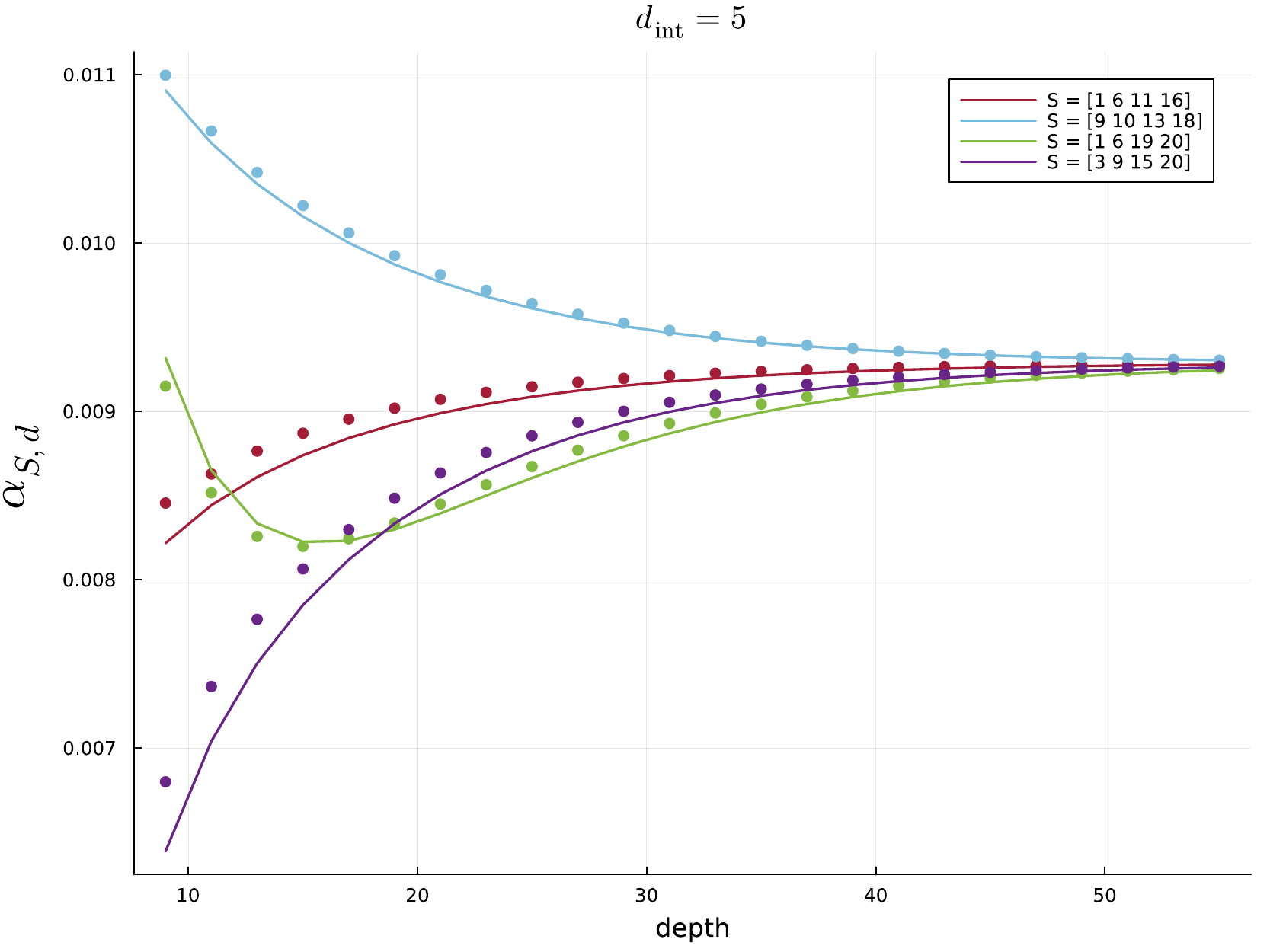}
    \caption{\centering Values of $\alpha_{S,d}$ and $\alpha_{S,d}'$ with the same $d_\text{int}$ and different $S$.}
    \label{fig: fix dint}
\end{figure}
\end{document}